\documentclass[sigconf,edbt]{acmart-edbt2023}

\def\BibTeX{{\rm B\kern-.05em{\sc i\kern-.025em b}\kern-.08em
    T\kern-.1667em\lower.7ex\hbox{E}\kern-.125emX}}

\usepackage{silence} 
\usepackage{booktabs} 
\usepackage[utf8]{inputenc}
\usepackage{graphicx}
\usepackage{textcomp}
\usepackage{amsmath}
\usepackage{amsfonts}
\usepackage{amsthm}
\usepackage{algorithm}
\usepackage{algpseudocode}
\usepackage{csquotes}
\usepackage{enumerate}
\usepackage{multirow}
\usepackage{multicol}
\usepackage{tabularx}
\usepackage{array}
\usepackage{xcolor}
\definecolor{light-gray}{gray}{0.95}
\usepackage{flushend}
\usepackage{comment}
\usepackage{mathrsfs}
\usepackage{gensymb}
\usepackage{tikz}
\usetikzlibrary{trees}
\usepackage{hyperref} 
\hypersetup{
    colorlinks=true,
    linkcolor=red,
    filecolor=magenta,      
    urlcolor=red,
    citecolor=red,
    }

\usepackage{caption}
\usepackage{subcaption}
\usepackage{nicefrac}
\usepackage{subfloat}
\usepackage{mathtools}
\DeclarePairedDelimiter\ceil{\lceil}{\rceil}

\setcopyright{rightsretained}

\acmDOI{}

\acmISBN{978-3-89318-092-9}

\acmConference[EDBT 2023]{26th International Conference on Extending Database Technology (EDBT)}{28th March-31st March, 2023}{Ioannina, Greece}
\acmYear{2023}

\def\eg{\emph{e.g.}}
\def\ie{\emph{i.e.}}
\def\cf{\emph{cf.}}
\def\etal{\textit{et al.}}
\def \ifempty#1{\def\temp{#1}\ifx\temp\empty}
\def \GRR#1{{\mathcal{M}_{\text{GRR}}\ifempty{#1}\else(#1)\fi}}
\def \rmH{\mathrm{H}}
\def \oneTo#1{{[1..#1]}}
\def \epsInf{{\epsilon_{\infty}}}
\def \epsIRR{{\epsilon_{\text{IRR}}}}
\def \epsOne{{\epsilon_1}}
\def \fBuckets{\mathtt{bucket}}
\def \calM{{\mathcal{M}}}

\def \bfv{{\mathbf{v}}}
\def \bfx{{\mathbf{x}}}

\def \valueLDP{{LDP on the users' values}}

\begin{document}
\sloppy

\title{Frequency Estimation of Evolving Data Under\\Local Differential Privacy}

\author{Héber H. Arcolezi$^{*}$}
\affiliation{
  \institution{Inria and École Polytechnique (IPP)}
}
\email{heber.hwang-arcolezi@inria.fr}

\author{Carlos Pinzón$^{*}$}
\affiliation{
  \institution{Inria and École Polytechnique (IPP)}
}
\email{carlos.pinzon@inria.fr}

\author{Catuscia Palamidessi}
\affiliation{
  \institution{Inria and École Polytechnique (IPP)}
}
\email{catuscia@lix.polytechnique.fr}

\author{Sébastien Gambs}
\affiliation{
  \institution{Université du Québec à Montréal, UQAM}
}
\email{gambs.sebastien@uqam.ca}
  
\thanks{$^{*}$ These are co-first authors that contributed equally to this work.}
\renewcommand{\shortauthors}{H.H. Arcolezi et al.}

\begin{abstract}
Collecting and analyzing evolving longitudinal data has become a common practice.
One possible approach to protect the users' privacy in this context is to use local differential privacy (LDP) protocols, which ensure the privacy protection of all users even in the case of a breach or data misuse.
Existing LDP data collection protocols such as Google's RAPPOR~\cite{rappor} and Microsoft's $d$BitFlipPM~\cite{microsoft} can have longitudinal privacy linear to the domain size $k$, which is excessive for large domains, such as Internet domains. 
To solve this issue, in this paper we introduce a new LDP data collection protocol for longitudinal frequency monitoring named LOngitudinal LOcal HAshing (LOLOHA) with formal privacy guarantees.
In addition, the privacy-utility trade-off of our protocol is only linear with respect to a reduced domain size $2\leq g\ll k$.
LOLOHA combines a domain reduction approach via local hashing with double randomization to minimize the privacy leakage incurred by data updates. 
As demonstrated by our theoretical analysis as well as our experimental evaluation, LOLOHA achieves a utility competitive to current state-of-the-art protocols, while substantially minimizing the longitudinal privacy budget consumption by up to $k/g$ orders of magnitude.
\end{abstract}

\maketitle

\section{Introduction} \label{sec:introduction}

Estimating histograms of evolving categorical data is a fundamental task in data analysis and data mining that requires collecting and processing data in a continuous manner.
A typical instance of such a problem is the online monitoring performed on software applications~\cite{Bittau2017}, for example for error reporting~\cite{Glerum2009}, to find commonly typed emojis~\cite{apple}, as well as to measure the users' system usage statistics~\cite{microsoft}.
However, the data collected can contain sensitive information such as location, health information, preferred webpage, etc. 
Thus, the direct collection and storage of users' raw data on a centralized server should be avoided to preserve their privacy. 
To address this issue, recent works have proposed several mechanisms satisfying Differential Privacy (DP)~\cite{Dwork2006DP,Dwork2006,dwork2014algorithmic} in the distributed setting in which an individual can directly randomize her own profile locally, referred to as Local DP (LDP)~\cite{first_ldp,Duchi2013,Duchi2013_b}.

One of the strengths of LDP is its simple trust model: since each user perturbs her data locally, user privacy is protected even if the server is malicious. 
For instance, some big tech companies have chosen to operate some of their applications in the local model, reporting the implementation of LDP protocols to collect statistics on well-known systems such as Google Chrome browser~\cite{rappor}, Apple iOS/macOS~\cite{apple}, and Windows 10 operating system~\cite{microsoft}).

Existing LDP protocols for frequency estimation typically focus on one-time computation~\cite{Feldman2022,Hadamard,tianhao2017,kairouz2016discrete,kairouz2016extremal,Cormode2021,Bassily2015,bassily2017practical}. 
However, considering both evolving data and the continuous monitoring together, pose a significant challenge under LDP guarantees. 
For instance, the naïve solution in which an LDP computation is repeated, will quickly increase the privacy loss leading to large values of $\epsilon$ due to the sequential composition theorem in DP~\cite{dwork2014algorithmic}. 
To tackle this issue, most state-of-the-art solutions relies on \textbf{memoization}~\cite{rappor,microsoft,erlingsson2020encode,Arcolezi2021,Arcolezi2021_allomfree}. 

Initially proposed by Erlingsson, Pihur, and Korolova~\cite{rappor}, the memoization-based RAPPOR protocol allows a user to memorize randomized versions of their true data and consistently reuse it when the same true value occurs. 
In addition, to improve privacy (\eg, minimize data change detection and/or tracking), the RAPPOR~\cite{rappor} protocol applies a second round of sanitization to the memoized value. 
However, the longitudinal privacy protection of RAPPOR only works if the underlying true value never or rarely changes (or changes in an uncorrelated fashion), which is unrealistic for evolving data (\eg, the number of seconds an application is used) as the privacy loss is proportional to the number of data changes, \ie, the domain size $k$ in the worst-case. 

To address this issue, Ding, Kulkarni, and Yekhanin~\cite{microsoft} have proposed a new LDP protocol named $d$BitFlipPM that improved memoization by mapping several values to the same randomized value. 
More precisely, $d$BitFlipPM partitions the original values into $b\leq k$ buckets (\eg, with equal widths), which allows close values to be mapped to the same bucket.
Afterwards, each user only samples $d \leq b$ buckets to minimize the number of bits to be randomized. 
Note that these two steps contributes to the information loss.
Another limitation of $d$BitFlipPM is the possibility of detecting data changes~\cite{Xue2022} on the fly since the true value will fall in a different bucket, there will be a higher probability of changing the randomization of the $d$ bits. 
Even if this only indicates that the user's value has changed, not what it was or is~\cite{rappor,microsoft}, there are still some privacy implications with respect to the type of inference an adversary can perform, especially if there are correlation patterns to be exploited~\cite{tang2017privacy,Naor2020}. 
Finally, $d$BitFlipPM's privacy loss can still be proportional to the number of bucket changes, \ie, the new domain size $b$ in the worst case. 

A different line of work has taken into account the infrequent data changes on the user side, hereafter referred to as \textbf{data change-based}~\cite{Joseph2018,erlingsson2019amplification,Xue2022,Ohrimenko2022}.
For instance, Joseph \etal~\cite{Joseph2018} have proposed a new LDP protocol THRESH for monitoring statistics (\eg, frequency) based on two sub-routines: voting and estimation, which requires splitting the privacy budget.
The main idea of THRESH is to update through voting the global estimate only when it becomes sufficiently inaccurate. 
However, privacy budget splitting under LDP guarantees is sub-optimal~\cite{tianhao2017,Arcolezi2021_allomfree,wang2019,Arcolezi2021,xiao2,erlingsson2020encode,Arcolezi2021_rs_fd}, which negatively impacts the data utility. 
Moreover, the authors in~\cite{erlingsson2019amplification,Ohrimenko2022} proposed the sanitization and report of data changes for frequency monitoring by assuming a limited number of data changes and longitudinal Boolean data, though it can be extend to larger domain. 
This leads to an accuracy that decays linearly (or sub-linearly) in the number of data changes.
Finally, in a recent work, Xue \etal~\cite{Xue2022} have proposed a new LDP protocol DDRM (Dynamic Difference Report Mechanism) based on difference trees.
However, DDRM assumes that the user's private sequence exhibit continuity (\ie, do not fluctuate significantly) and was mainly designed for longitudinal Boolean data. 
Besides, DDRM requires a privacy budget allocation scheme that depends on the number of data collections as well as to split the privacy budget when extending to a larger domain (\ie, sub-optimal). 

\textbf{Main contributions.} 
In this paper, we address the limitations of memoization-based protocols~\cite{rappor,microsoft,Arcolezi2021_allomfree,erlingsson2020encode} without imposing any restriction on the number of data changes and/or on the number of data collections as in data change-based protocols~\cite{Joseph2018,erlingsson2019amplification,Xue2022,Ohrimenko2022}. 
More precisely, we propose a novel LDP protocol with formal privacy guarantees for longitudinal frequency estimation of evolving counter (or categorical) data.

Our protocol, hereafter named LOngitudinal LOcal HAshing (LOLOHA), combines a domain reduction approach through local hashing~\cite{Bassily2015,tianhao2017} with the memoization solution of RAPPOR using two rounds of sanitization~\cite{rappor,Arcolezi2021_allomfree}.
The main strength of LOLOHA is that the longitudinal privacy-utility trade-off is linear only on the new (reduced) domain size $g$, in which $2\leq g \ll k$ is a tunable hyper-parameter.
This way, the worst-case longitudinal privacy loss of LOLOHA has a significant $k/g$ or $b/g$ decrease factor in comparison with RAPPOR and $d$BitFlipPM, respectively.

Indeed, LOLOHA can be tuned for strong longitudinal privacy by selecting $g=2$ (BiLOLOHA protocol). 
To maximize LOLOHA's utility, we also find the optimal $g$ value (OLOLOHA protocol). 
Experimental evaluations demonstrate the effectiveness of LOLOHA with respect to the quality of frequency estimates, in addition to substantially minimizing the longitudinal privacy loss.

We also show why LDP is generally impossible to achieve when data is longitudinal, which motivates a definition of privacy that better suits the longitudinal scenario.
This is in opposition with the common and mathematically equivalent path in the literature of claiming a protocol to be LDP but assuming that the evolving data is uncorrelated or constant in time, which we believe not be realistic in real-life deployments.

In summary, the main contributions of this paper are three-fold: 
\begin{itemize}
    \item We propose the LOLOHA protocol for longitudinal frequency monitoring under LDP guarantees.
    \item We prove the longitudinal privacy and accuracy guarantees of LOLOHA through theoretical analysis and compare it to existing protocols.
    \item We show the performance of LOLOHA numerically and experimentally, using both real-world and synthetic datasets.
\end{itemize}

\textbf{Outline.} The remainder of this paper is organized as follows. 
First, in Section~\ref{sec:preliminaries}, we provide the problem definition and review LDP and existing longitudinal LDP protocols.
Next, we present and analyze our LOLOHA protocols in Section~\ref{sec:loloha}. 
In Section~\ref{sec:theoretical_comparison}, we give a theoretical comparison of LOLOHA and state-of-the-art LDP protocols before presenting and interpreting the experimental results in Section~\ref{sec:results_discussion}. 
Finally, in Section~\ref{sec:rel_work}, we review related work before concluding with future perspectives in Section~\ref{sec:conclusion}.

\section{Preliminaries} \label{sec:preliminaries}

In this section, we present the problem considered and we review the LDP privacy model and relevant protocols.

\noindent \textbf{Notation.} For denoting sets, we will use italic uppercase letters $V$, $U$, etc, and we write $\oneTo{n}=\{1,\ldots,n\}$. 
For a vector $\textbf{x}$ (bold lowercase letters), $\bfx_i$ represents the value of its $i$-th coordinate. 
Finally, we denote randomized protocols as $\calM$.

\subsection{Problem Statement} 
\label{sub:problem}

We consider the situation in which a server collects data from a distributed group of users while requiring the protection of privacy for each user, through LDP.
The server collects sanitized data over time from each member of the group with respect to a fixed discrete random variable (\eg, daily usage of a mobile application). 
Its objective is to estimate the true frequencies, or histograms, of the random variable as well as its evolution over time.
We aim to provide the server with an optimized combination of two algorithms: one for the users, who must sanitize locally their data before sending it, and another for the server, which wants to aggregate data and perform the estimation accurately.

Formally, there are $n$ users $U=\{u_1,\ldots,u_n\}$ and a random variable taking values in a set $V$ of size $k$ with true frequencies $\{f(v)\}_{v\in V}$, which may vary over time. 
Each user $u \in U$, holds a private sequence of values $\textbf{v}^{(u)}= \left [ v^{(u)}_{1},v^{(u)}_{2},\ldots,v^{(u)}_{\tau} \right ]$, in which $v_t^{(u)}$ represents the discrete value $v \in V$ of user $u$ at time step $t \in \oneTo{\tau}$.
At each time step $t$, upon collecting the sanitized values of all $n$ users, the server will estimate a $k$-bins histogram $\{\hat f(v)\}_{v\in V}$ in a way that minimizes the \textit{Mean Squared Error} (MSE) with respect to $\{f(v)\}_{v\in V}$.
For all the algorithms presented hereafter, the estimation $\hat{f}(v)$ is unbiased (\ie, $\mathbb{E}(\hat{f}(v))=f(v)$).
As a consequence, the MSE is equivalent to the variance as:

\begin{equation*} \label{eq:mse_var}
    \mathrm{MSE}
    = \frac{1}{|V|} \sum_{v \in V} \mathbb{E} \left[ \left( \hat{f}(v) - f(v) \right)^2 \right ] 
    = \frac{1}{|V|} \sum_{v \in V} \mathbb{V}[\hat{f}(v)] \textrm{.}
\end{equation*}

\subsection{Local Differential Privacy} \label{sub:ldp}

\noindent \textbf{Privacy model.} In this paper, we use LDP (Local Differential Privacy)~\cite{first_ldp,Duchi2013,Duchi2013_b} as the privacy model considered, which is formally defined as follows.

\begin{definition}[$\epsilon$-Local Differential Privacy]\label{def:ldp} A randomized algorithm ${\calM}$ satisfies $\epsilon$-local-differential-privacy ($\epsilon$-LDP), where $\epsilon>0$, if for any pair of input values $v_1, v_2 \in Domain(\calM)$ and any possible output $x'$ of ${\calM}$:

\begin{equation*} \label{eq:ldp}
    \Pr[{\calM}(v_1) = x'] \leq e^\epsilon \cdot \Pr[{\calM}(v_2) = x']  \textrm{.}
\end{equation*}
\end{definition}

In essence, LDP guarantees that it is unlikely for the data aggregator to reconstruct the input data.
The privacy loss $\epsilon$ controls the privacy-utility trade-off for which lower values of $\epsilon$ result in tighter privacy protection. 
Similar to central DP, LDP also has several fundamental properties, such as robustness to post-processing and composition~\cite{dwork2014algorithmic}. 

\begin{proposition}[Post-Processing~\cite{dwork2014algorithmic}]
\label{prop:post-processing}
If $\calM$ is $\epsilon$-LDP, then $f (\calM)$ is also $\epsilon$-LDP for any function $f$.
\end{proposition}

\begin{proposition}[Sequential Composition~\cite{dwork2014algorithmic}]
\label{prop:seq_comp}
Let $\calM_t$ be $\epsilon_t$-LDP mechanism, for $t \in [\tau]$. 
Then, the sequence of outputs $[\calM_1(v),\ldots, \calM_{\tau}(v)]$ is $\sum_{t=1}^{\tau} \epsilon_t$-LDP.
Moreover, if $\calM$ is an $\epsilon$-LDP mechanism and $\bfv$ is a finite sequence of $k$ values, then the sequence of outputs $[\calM(v_1),\ldots,\calM(v_k)]$ is $k \epsilon$-LDP.
\end{proposition}

\subsection{LDP Frequency Estimation Protocols} \label{sub:fo}

In this section, we review five state-of-the-art LDP frequency estimation protocols, which are often used as building blocks for more complex tasks (\eg, heavy hitter estimation~\cite{Bassily2015,bassily2017practical}, machine learning~\cite{Chamikara2020}, and private frequency monitoring~\cite{microsoft,rappor,Arcolezi2021_allomfree}).

\subsubsection{Generalized Randomized Response (GRR)}  \label{sub:grr}

The GRR~\cite{kairouz2016discrete,kairouz2016extremal} protocol generalizes the Randomized Response (RR) technique proposed by Warner~\cite{Warner1965} for $k \geq 2$ while satisfying LDP.

Fix a parameter $\epsilon>0$ and let $p\coloneqq\frac{e^{\epsilon}}{e^{\epsilon}+k-1}\in (0,1)$ in which $k=|V|$.
For each $v\in V$, let $\eta_{\neq v}\in V$ be a uniform (\emph{i.e.}, exogenous noise) random variable over $V\setminus\{v\}$.
We let $\GRR{}: V \to V$ be the random variable given by:

\begin{equation*}
    \GRR{v;\epsilon} \coloneqq \begin{cases}
        v , &\textrm{w.p. } p \\
        \eta_{\neq v}, &\textrm{w.p. } 1-p \textrm{.}
    \end{cases}
\end{equation*}

This protocol satisfies $\epsilon$-LDP, because $\frac{p}{q}=e^{\epsilon}$~\cite{kairouz2016discrete}, in which $q\coloneqq \nicefrac{(1-p)}{(k-1)}$ determines the probability of the response being any fixed noise value different of $v$.
To estimate the normalized frequency of $v \in V$, one counts how many times $v$ is reported, expressed as $C(v)$, and then computes:

\begin{equation}\label{eq:est_pure}
    \hat{f}(v) = \frac{C(v) - nq}{n(p - q)} \textrm{,}
\end{equation}

\noindent in which $n$ is the total number of users. 
In~\cite{tianhao2017}, it was proven that Eq.~\eqref{eq:est_pure} is an unbiased estimator (\ie, $\mathbb{E}(\hat{f}(v))=f(v)$).

\subsubsection{Local Hashing (LH)}  \label{sub:lh_protocols}

LH protocols~\cite{tianhao2017} can handle a large domain size $k$ by first using hash functions to map an input value to a smaller domain of size $g\geq 2$ (typically $g \ll k)$, and then applying GRR to the hashed value. 

Fix $\epsilon>0$ and let $\GRR{}: \oneTo{g}\to\oneTo{g}$ be the GRR mechanism with parameter $\epsilon$ and assuming the input-output domain to be $\oneTo{g}$ instead of $V$, so that the size is $g$ instead of $k$.
In local hashing, each user selects at random a hashing function $\rmH$ from a family of universal hash functions, and reports the pair $\langle \rmH, \GRR{x;\epsilon} \rangle$, in which $x=\rmH(v)$.

The hash values will remain unchanged with probability $p=\frac{e^{\epsilon}}{e^{\epsilon}+g-1}$ and switch to any different fixed value in $\oneTo{g}$ with probability $q=\frac{1}{e^{\epsilon}+g-1}$.
This means that for each hash value $x\in\oneTo{g}$, it holds that:
\begin{equation*}
    \Pr[\GRR{\rmH(v);\epsilon} = x] = \begin{cases}
        p, &\textrm{if } x = \rmH(v)\\
        q, &\textrm{otherwise.}
    \end{cases}
\end{equation*}

Let $\langle \rmH^u, x^u \rangle$ be the report from user $u\in U$. The server can obtain the unbiased estimation of $v \in V$, with Eq.~\eqref{eq:est_pure} by setting $q=\frac{1}{g}$ and $C(v) = |\{u\in U \mid \rmH^u(v) = x^u \}|$~\cite{tianhao2017}. 

The authors in~\cite{tianhao2017} describe two LH protocols that differ on how $g$ is selected: (1) Binary LH (BLH) that selects $g=2$ and (2) Optimal LH (OLH) that selects $g=\lfloor e^{\epsilon} + 1 \rceil$ (rounded to closest integer). 

\subsubsection{Unary Encoding (UE)} \label{sub:ue_protocols}

UE protocols interpret the user's input $v \in V$, as a one-hot $k$-dimensional vector.
More precisely, $\textbf{x}=\texttt{UE}(v)$ is a binary vector with only the bit at the position corresponding to $v$ set to 1 and the other bits set to 0. 
The perturbation function of UE protocols randomizes the bits from $\textbf{x}$ independently with probabilities:

\begin{equation}  \label{eq:ue_parameters}
    \forall{i \in [k]} : \quad \Pr[\textbf{x}_i'=1] =\begin{cases} p, \textrm{ if } \textbf{x}_i=1 \textrm{,} \\ q, \textrm{ if } \textbf{x}_i=0 \textrm{.}\end{cases}
\end{equation}

Afterwards, the client sends $\textbf{x}'$ to the server. 
The authors in~\cite{tianhao2017} describe two UE protocols that depend on the parameters $p$ and $q$ in Eq.~\eqref{eq:ue_parameters}: (1) Symmetric UE (SUE)~\cite{rappor}, which selects $p=\frac{e^{\epsilon/2}}{e^{\epsilon/2}+1}$ and $q=\frac{1}{e^{\epsilon/2}+1}$ such that $p+q=1$, and (2) Optimal UE (OUE), which selects $p=\frac{1}{2}$ and $q=\frac{1}{e^{\epsilon}+1}$. 

The estimation method used in Eq.~\eqref{eq:est_pure} applies equally to both UE protocols, in which $C(v)$ represents the number of times the bit corresponding to $v$ has been reported.
Last, both SUE and OUE protocols satisfy $\epsilon$-LDP for $\epsilon = ln\left( \frac{p(1-q)}{(1-p)q} \right )$~\cite{tianhao2017}.

\subsection{Existing Longitudinal LDP Frequency Estimation Protocols} \label{sub:long_fo}

For privately monitoring the frequency of values of a population, the simplest way is that each user adds independent fresh noise to $v$ in each data collection $t\in \oneTo{\tau}$ following one of the LDP protocols described in the previous section. 
However, this solution is vulnerable to ``averaging attacks'' in which an adversary can estimate the true value from observing multiple randomized versions of it. 
To avoid this averaging attack, the memoization approach~\cite{rappor} was designed to enable longitudinal collections through memorizing a randomized version of the true value $v$ and consistently reusing it~\cite{microsoft,Arcolezi2021} or reusing it as the input to a second round of sanitization (\ie, chaining two LDP protocols)~\cite{rappor,Arcolezi2021_allomfree,erlingsson2020encode}. 
The next four subsections describe state-of-the-art memoization-based protocols.

\subsubsection{RAPPOR Protocol} \label{sub:rappor}

The utility-oriented version of RAPPOR~\cite{rappor} is based on the SUE protocol, which encodes the user's input $v \in V$ as a $k$-dimensional bit-vector and randomizes each bit independently.
More specifically, for each value $v \in V$, the user encodes $\textbf{x}=\texttt{UE}(v)$ and randomizes $\textbf{x}$ as follows:

\textit{Step 1. Permanent RR (PRR):} Memoize $\textbf{x}'$ such that:

\begin{equation*}  
    \forall{i \in [k]} : \quad \Pr[\textbf{x}_i'=1] =\begin{cases} p_1=\frac{e^{\epsInf/2}}{e^{\epsInf/2}+1}, \textrm{ if } \textbf{x}_i=1 \textrm{,} \\ q_1=\frac{1}{e^{\epsInf/2}+1}, \textrm{ if } \textbf{x}_i=0 \textrm{,}\end{cases}
\end{equation*}

\noindent in which $p_1$ and $q_1$ control the level of longitudinal $\epsInf$-LDP for $\epsInf = ln\left( \frac{p_1(1-q_1)}{(1-p_1)q_1} \right )$~\cite{rappor}. 
This step is carried out only once for each value $v \in V$ that the user has. 
Thus, the value $\textbf{x}'$ shall be reused as the basis for all future reports of $v$. 

\textit{Step 2. Instantaneous RR (IRR):} Generate $\textbf{x}''$ such that:

\begin{equation*}  
    \forall{i \in [k]} : \quad \Pr[\textbf{x}_i''=1] =\begin{cases} p_2, \textrm{ if } \textbf{x}_i'=1 \textrm{,}\\ q_2, \textrm{ if } \textbf{x}_i'=0 \textrm{.}\end{cases}
\end{equation*}

This second step is carried out each time $t\in\oneTo{\tau}$ a user report the value $v$. RAPPOR's deployment selected $p_2=0.75$ and $q_2=0.25$~\cite{rappor,tianhao2017} (\ie, also symmetric). 
The RAPPOR protocol that chains two SUE protocols is referred to as L-SUE in~\cite{Arcolezi2021_allomfree,Arcolezi2022_multi_freq_ldpy}. 
We provide the calculation of parameters $p_2$ and $q_2$ in the repository~\cite{artifact_loloha}.
Note that $\epsInf$ corresponds to an upper bound for each value $v$ as $t \rightarrow \infty$. 
The privacy guarantees of the IRR step degrade according to the number of reports $t \in \oneTo{\tau}$~\cite{rappor,erlingsson2020encode}.

With two rounds of sanitization, each consisting of an LDP protocol parametrized with $p, q$, the unbiased estimator in Eq.~\eqref{eq:est_pure} is now extended to~\cite{Arcolezi2021_allomfree,rappor}: 

\begin{equation}\label{eq:est_longitudinal}
    \hat{f}_L(v) = \frac{\frac{C(v) - nq_2}{(p_2-q_2)} - nq_1}{n(p_1-q_1)} = \frac{C(v) - nq_1(p_2-q_2) - nq_2}{n(p_1-q_1)(p_2-q_2)} \textrm{,}
\end{equation}

\noindent in which $p_1$ and $q_1$ are the parameters of the LDP protocol used in the first step while $p_2$ and $q_2$ are the parameters of the LDP protocol used in the second step.

In~\cite{Arcolezi2021_allomfree}, it was proven that Eq.~\eqref{eq:est_longitudinal} is an unbiased estimator (\ie, $\mathbb{E}(\hat{f}_L(v))=f(v)$) and that for any value $v \in V$, the variance $\mathbb{V}$ of the estimator $\hat{f}_L(v)$ in Eq.~\eqref{eq:est_longitudinal} is:

\begin{equation}\label{var:longitudinal}
\begin{gathered}
    \mathbb{V}[\hat{f}_L(v)]  = \frac{\gamma (1-\gamma)}{n (p_1-q_1)^2 (p_2-q_2)^2} \textrm{, where} \\
    \gamma = f(v) \left( 2 p_{1} p_{2} - 2 p_{1} q_{2} + 2 q_{2} - 1 \right) + p_{2} q_{1} + q_{2} (1 - q_{1})  \textrm{.}
\end{gathered}
\end{equation}

In this paper, we will use the \textit{approximate variance} $\mathbb{V}^*$, in which $f(v)=0$ in Eq.~\eqref{var:longitudinal}, which gives:

\begin{equation}\label{var:aprox_longitudinal}
    \mathbb{V}^*\left[\hat{f}_L(v)\right]  = \frac{\left(p_{2} q_{1} - q_{2} \left(q_{1} - 1\right)\right) \left(- p_{2} q_{1} + q_{2} \left(q_{1} - 1\right) + 1\right)}{n (p_1-q_1)^2 (p_2-q_2)^2}  \textrm{.}
\end{equation}

Therefore, one can obtain the RAPPOR approximate variance $\mathbb{V}^*[\hat{f}_{\textrm{RAPPOR}}(v)]$ by replacing the resulting $p_1,q_1,p_2,q_2$ parameters into Eq.~\eqref{var:aprox_longitudinal}.

\subsubsection{Optimized Longitudinal UE Protocol} \label{sub:l_osue}

The authors in~\cite{Arcolezi2021_allomfree} analyzed all four combinations between OUE and SUE in both PRR and IRR steps.
The optimized protocol named L-OSUE chains the OUE protocol (PRR step) and the SUE protocol (IRR step). 
Thus, for each value $v \in V$, the user encodes $\textbf{x}=\texttt{UE}(v)$ and randomizes $\textbf{x}$ as follows:

\textit{Step 1. PRR:} Memoize $\textbf{x}'$ such that:

\begin{equation*}  
    \forall{i \in [k]} : \quad \Pr[\textbf{x}_i'=1] =\begin{cases} p_1=\frac{1}{2}, \quad \quad \textrm{ if } \textbf{x}_i=1 \textrm{,} \\ q_1=\frac{1}{e^{\epsInf}+1}, \textrm{ if } \textbf{x}_i=0 \textrm{,}\end{cases}
\end{equation*}

\noindent in which $p_1$ and $q_1$ control the level of longitudinal $\epsInf$-LDP as $e^{\epsInf} = \frac{p_1(1-q_1)}{q_1(q-p_1)}$~\cite{rappor,Arcolezi2021_allomfree}. 
The value $\textbf{x}'$ shall be reused as the basis for all future reports when the real value is $v$.

\textit{Step 2. IRR:} Generate $\textbf{x}''$ such that:

\begin{equation*}  
    \forall{i \in [k]} : \quad \Pr[\textbf{x}_i''=1] =\begin{cases} p_2 , \quad \quad \quad \quad \textrm{ if } \textbf{x}_i'=1 \textrm{,}\\ q_2= 1 - p_2, \textrm{ if } \textbf{x}_i'=0 \textrm{.}\end{cases}
\end{equation*}

\noindent in which $p_2= \frac{e^\epsInf e^\epsOne - 1}{e^\epsInf - e^\epsOne + e^{\epsInf + \epsOne} - 1}$ and $\textbf{x}''$ is the report to be sent to the server.
Let $p_s=\Pr[ \textbf{x}_i'' =1 | \textbf{x}_i =1] = p_1 p_2 + (1 - p_1) q_2$ and $q_s =\Pr[ \textbf{x}_i'' =1 | \textbf{x}_i =0 ] = q_1 p_2 + (1 - q_1) q_2$. 
For the first report, the L-OSUE protocol satisfies $\epsOne$-LDP as $e^{\epsOne}=\frac{p_s(1-q_s)}{(1-p_s)q_s}$~\cite{Arcolezi2021_allomfree,rappor}. 

Similar to RAPPOR, the estimated frequency $\hat{f}_{\textrm{L-OSUE}}(v)$ that a value $v\in V$ occurs, can be computed using Eq.~\eqref{eq:est_longitudinal}. 
One can also obtain the L-OSUE approximate variance $\mathbb{V}^*[\hat{f}_{\textrm{L-OSUE}}(v)]$ by replacing the resulting $p_1,q_1,p_2,q_2$ parameters into Eq.~\eqref{var:aprox_longitudinal}.

\subsubsection{Longitudinal GRR (L-GRR)} \label{sub:l_grr}

The L-GRR~\cite{Arcolezi2021_allomfree} protocol chains GRR in both PRR and IRR steps.
Therefore, for each value $v \in V$, the user randomizes $v$ as follows:

\textit{Step 1. PRR:} Memoize $x'$ such that:
\begin{equation*}
x'=
\begin{cases}
  v, & \textrm{w.p.}\ p_1 = \frac{e^{\epsInf}}{e^{\epsInf}+k-1} \textrm{,}\\
  \tilde v \in V \setminus \{v\}, & \textrm{w.p.}\ q_1=\frac{1-p_1}{k-1} \textrm{,} \\
\end{cases}
\end{equation*}

\noindent in which $p_1$ and $q_1$ control the level of longitudinal $\epsInf$-LDP as $e^{\epsInf}=\frac{p_1}{q_1}$~\cite{kairouz2016discrete,Arcolezi2021_allomfree}. The value $x'$ shall be reused as the basis for all future reports on the real value $v$.
  
\textit{Step 2. IRR:} Generate a report $x''$ such that:
\begin{equation*}\label{eq:perm}
x''=
\begin{cases}
  x', & \textrm{w.p.}\ p_2 \textrm{,}\\
  \tilde x \in V \setminus \{x'\}, & \textrm{w.p.}\ q_2=\frac{1-p_2}{k-1}  \textrm{,}\\
\end{cases}
\end{equation*}
\noindent in which $p_2 = \frac{e^{\epsInf + \epsOne} - 1}{- k e^\epsOne + \left(k - 1\right) e^{\epsInf} + e^\epsOne + e^{\epsOne + \epsInf} - 1}$ and $x''$ is the report to be sent to the server.
Let $p_s=\Pr \left[ x''=v | v \right] = p_1 p_2 + q_1 q_2$ and $q_s = \Pr[ x''=v | \tilde{v} \in V \setminus \{v\} ] = p_1 q_2 + q_1 p_2$. For the first report, the L-GRR protocol satisfies $\epsOne$-LDP since $e^\epsOne=\frac{p_s}{q_s}$~\cite{Arcolezi2021_allomfree}.

The estimated frequency $\hat{f}_{\textrm{L-GRR}}(v)$ that a value $v$ occurs can also be obtained using Eq.~\eqref{eq:est_longitudinal}. 
Besides, one can compute the L-GRR approximate variance $\mathbb{V}^*[\hat{f}_{\textrm{L-GRR}}(v)]$ by replacing the resulting $p_1,q_1,p_2,q_2$ parameters into Eq.~\eqref{var:aprox_longitudinal}.

\subsubsection{dBitFlipPM Protocol} \label{sub:microsoft}

The $d$BitFlipPM~\cite{microsoft} protocol was proposed to improve the memoization solution of RAPPOR~\cite{rappor} by mapping several true values to the same noisy response at the cost of losing information due to generalization. 
This is done by first partitioning the original domain $V$ into $b$ buckets (\ie, new domain size $2\leq b \leq k$) using a function $\fBuckets: V\to \oneTo{b}$, such that close values will fall into the same bucket. 
Next, each user randomly draws $d$ bucket numbers without replacement from $\oneTo{b}$, denoted by $j_1, j_2, \ldots, j_d$, and fixes them for all future data collections.
Then, for each $v \in V$, the user sends a sanitized vector $\textbf{x}'=\left [ (j_1, x_{j_1}), \ldots, (j_d, x_{j_d}) \right]$ parameterized with the privacy guarantee $\epsInf$ as follows:

\begin{equation*}
    \forall{l \in \oneTo{d}} : \Pr[x_{j_l}=1] = \begin{cases}
        p=\frac{e^{\epsInf/2}}{e^{\epsInf/2}+1}, \textrm{ if } \fBuckets(v)=j_l\\
        q=\frac{1}{e^{\epsInf/2}+1}, \textrm{ if } \fBuckets(v)\neq j_l
    \end{cases}
    \textrm{.}
\end{equation*}

In other words, users inform the server which bits are sampled as well as their perturbed values, but the server does not receive any information about the remaining $b-d$ bits.
The server can estimate the number of times each bucket in $\oneTo{b}$ has been reported with Eq.~\eqref{eq:est_pure} by replacing $n$ with $\frac{nd}{b}$ as each user only sampled $d$ bits among $b$ buckets.

In contrast to RAPPOR, there is no second round of sanitization, which means the user runs $d$BitFlipPM with $\epsInf$-LDP for all $b$ buckets, with randomization applied to the $d$ fixed bits $j_1, j_2,\ldots,j_d$ and memoizes the response. 
This approach adds uncertainty to the real value because multiple (close) values will be mapped to the same bucket.
The highest protection is given when $d=1$~\cite{microsoft}, which will minimize the chances (to some extent) of detecting high data changes.

\section{LOLOHA} \label{sec:loloha}

In this section, we introduce our LOLOHA (Longitudinal Local Hashing) protocol for frequency monitoring throughout time under LDP constraints, and we analyze its utility and privacy.

The privacy analysis of longitudinal protocols requires special treatment because, since they are stateful, they cannot be modeled as mechanisms mapping values into values, but rather sequences into sequences.
This makes the LDP constraint too strong in the long term as shown in the following theorem.

\begin{theorem}\label{thm:ldpIsImpossible}
(LDP cannot be satisfied when $\tau\to\infty$)
Consider a randomized longitudinal mechanism $\vec\calM:\oneTo{n}^\tau\to\oneTo{m}^\tau$ mapping an input sequence $X_1,\ldots,X_\tau$ to an output sequence $Y_1,\ldots,Y_\tau$, for some positive integer $\tau$.
For the sake of utility of each reported value $Y_t$ ($0$-LDP means total detriment of utility), assume some negligible but positive fixed $\alpha>0$ such that the mechanism for generating $Y_t$ from $X_t$ and the history $X_1,Y_1,\ldots,X_{t-1},Y_{t-1}$ is not $\alpha$-LDP.
If $\tau \geq \epsilon/\alpha$ then $\vec\calM$ is not $\epsilon$-LDP.
\end{theorem}
\begin{proof}
Let $y_1 = \arg\max_y \frac{\max_x P(X_1=x| Y_1=y)}{\min_x P(X_1=x| Y_1=y)}$, and call $x_1^+$ and $x_1^-$ to the values that respectively maximize and minimize $P(X_1=x|Y_1=y_1)$.
By the minimal utility assumption, $\frac{p(x_1^+|y_1)}{p(x_1^-|y_1)} > e^\alpha$.

Let $y_2 = \arg\max_y \frac{\max_x P(X_2=x| Y_2=y, Y_1=y_1, X_1=x_1)}{\min_x P(X_2=x | Y_2=y, Y_1=y_1, X_1=x_1)}$, and call $x_2^+$ and $x_2^-$ to the values that respectively maximize and minimize $P(X_2=x | Y_2=y, Y_1=y_1, X_1=x_1)$.
Since the output values of the mechanism are reported one by one in temporal order, we have $p(x_1,x_2|y_1,y_2)=p(x_1|y_1)p(x_2|y_2,y_1,x_1)$, hence by the minimal utility assumption and the first step, $\frac{p(x_1^+,x_2^+|y_1,y_2)}{p(x_1^-,x_2^-|y_1,y_2)}= \frac{p(x_1^+|y_1)p(x_2^+|y_2,y_1,x_1^+)}{p(x_1^-|y_1)p(x_2^-|y_2,y_1,x_1^-)} > e^{2\alpha}$.

Repeating this process inductively yields three sequences $y_i$, $x_i^+$ and $x_i^-$ of length $\tau$ such that $\frac{p(x_1^+,\ldots,x_\tau^+|y_1,\ldots,y_\tau)}{p(x_1^-,\ldots,x_\tau^-|y_1,\ldots,y_\tau)} > e^{\tau \alpha}$.

This makes it impossible for the mechanism to be $\epsilon$-LDP for any $\tau\geq \epsilon/\alpha$.
\end{proof}

For instance, assume that a user has a secret sequence $\bfv=[1, 1, 1, 3, 1, 2, 1,1,3]$ ($\tau=9$ time steps), and reports $\vec \calM(\bfv)\coloneqq [\calM(v_1), \ldots, \calM(v_9)]$, in which $\calM$ is the memoization mechanism $(1\mapsto 2; 2\mapsto2; 3\mapsto 3)$ that reuses the sanitized report.
The server receives $[2,2,2,3,2,2,2,2,3]$, hence some time-related patterns in the sequence are exposed, but the memoization protects the uncertainty about the user actual values.
As the sequence size grows, the vectorized memoization mechanism $\vec\calM$ that processes temporal data continues to protect the values indefinitely, but fails to satisfy LDP.
For this reason, we introduce the following relaxed definition of privacy for longitudinal mechanisms.

\begin{definition}[Longitudinal LDP] \label{def:ldp_memo}

For a longitudinal memoizing mechanism $\calM:A^\tau\to B^\tau$, in which $A=\oneTo{k}$, let $\calM^\star$ denote a mechanism that takes as input a permutation $x$ of $A$ and outputs $\calM^\star(x):= x''$ by shuffling the $k$ entries of $x$, yielding $x'$, and letting $x''_i:=\calM(x'_i)$ for each $i=1..k$, sequentially.
$\calM$ is said to be $\epsilon$-\valueLDP{} iff $\calM^\star$ is $\epsilon$-LDP.
\end{definition}

Definition~\ref{def:ldp_memo} discards all information contained in time correlation by shuffling the input and aggregates the total privacy loss after all input values have been memoized.
Moreover, Definition~\ref{def:ldp_memo} corresponds to the total privacy budget that will be consumed for sanitizing all the values of the user.

Previous influential works, such as RAPPOR~\cite{rappor} and $d$BitFlipPM~\cite{microsoft}, handle the negative consequences of Theorem~\ref{thm:ldpIsImpossible} implicitly by assuming that the data values (or buckets) never change or change in an uncorrelated manner.
We consider the former to be unrealistic and the latter is insufficient to guarantee LDP, though it makes users indistinguishable.
In this paper, we privileged Definition~\ref{def:ldp_memo} over extreme assumptions on the data to be able to explain at least what is actually being protected by the mechanism when the assumptions do not hold.
Hence, we present long term guarantees in terms of \valueLDP{}, but also, single-report LDP guarantees, as done in the literature, which are equivalent to LDP assuming constant values.

\subsection{Overview of LOLOHA}

LOLOHA is inspired by the strengths of RAPPOR~\cite{rappor} (double sanitization to minimize data change detection) and $d$BitFlipPM~\cite{microsoft} (several values are mapped to the same randomized value) protocols. 
More precisely, LOLOHA is based on LH for the PRR step to satisfy $\epsInf$-LDP (upper bound), which significantly reduces the domain size. 
Thus, the user will uniformly choose at random a universal hash function $\rmH$ that maps the original domain $V \rightarrow \oneTo{g}$, with $g\geq 2$ typically much smaller than $k=|V|$.
Indeed, given a general (universal) family of hash functions $\mathscr{H}$, each input value $v \in V$ is hashed into a value in $\oneTo{g}$ by hash function $\rmH \in \mathscr{H}$, and the universal property requires: 

\begin{equation*}
    \forall{v_1, v_2} \in V, v_1 \neq v_2 : \quad \underset{\rmH \in \mathscr{H}}{\Pr}\left[\rmH(v_1) = \rmH(v_2)  \right] \leq \frac{1}{g} \textrm{.}
\end{equation*}

In other words, approximately $k/g$ values $v \in V$ can be mapped to the same hashed value $\rmH(v)$ in $\oneTo{g}$ due to collision. 
After the hashing step, to satisfy $\epsInf$-LDP, the user invokes the GRR protocol to the hashed value $x=\rmH(v)$ and memoizes the response $x'=\GRR{x; \epsInf}$.
Then, the value $x'$ will be reused as the basis for all future reports on the \textbf{hashed value} $x$, which supports all values in set $X_{H}= \{v \in V \mid \rmH(v) = x\}$.
The intuition is that the user only leaks $\epsInf$ for each hashed value $x \in \oneTo{g}$ as they support all values $v\in V$ that collide to $x=\rmH(v)$.
Notice that instead of memoization, users could also pre-compute the mapping for each input value. 
These two methods would be equivalent in terms of the functionality provided. 

Moreover, in contrast with the $d$BitFlipPM protocol in which only close values are mapped to the same bucket, any two values in $V$ can collide with probability at most $1/g$.
Therefore, even if the user's value changes periodically, correlated or in a abrupt manner, there will still be uncertainty on the actual value $v$. 
However, with only this PRR step, it would be possible to detect some of the data changes due to the randomization of a different hash value. 
Therefore, LOLOHA also requires the user to apply a second round of sanitization (\ie, IRR step) to the memoized values $x'$ with the GRR protocol such that the first report satisfies $\epsOne$-LDP, for some chosen positive $\epsOne<\epsInf$.

\subsection{Client-Side of LOLOHA}  \label{sub:client_side}

Algorithm~\ref{alg:loloha} displays the pseudocode of LOLOHA on the client-side, which receives as input: the true sequence of values $\textbf{v}= \left [ v_{1},v_{2},\ldots,v_{\tau} \right ]$ of the user that is running the code, a universal family $\mathscr{H}$ of hash functions $\rmH:V \to \oneTo{g}$, and the constants $\epsOne, \epsInf$, with $0< \epsOne <\epsInf$, that represent respectively the leakage of the first report and the maximal longitudinal leakage.

\begin{algorithm}[!h]
\caption{Client-Side of LOLOHA.}
\label{alg:loloha}
\begin{algorithmic}[1]
\Statex \textbf{Input:} User longitudinal values $\left [ v_{1},v_{2},\ldots,v_{\tau} \right ]$, family $\mathscr{H}$ of hash functions and constants $0<\epsOne<\epsInf$.
\Statex \textbf{Output:} None. Sends data to server during execution.
\raisebox{-0.8em}{\,}

\State $\rmH \gets_R \mathscr{H}$\Comment{Hash function chosen at random}

\State Send $H$.

\State $\epsIRR \gets \ln\left(\frac{e^{\epsInf+\epsOne}-1}{e^\epsInf - e^\epsOne}\right)$

\State \textbf{for} each time $t \in \oneTo{\tau}$ \textbf{do}:

\State \hskip1em $x \gets \rmH(v_t)$. \Comment{Hash step}

\State \hskip1em \textbf{if} $x$ \textit{is not} memoized \textbf{then}:

\State \hskip2em $x' \gets \GRR{x; \epsInf}$ over $\oneTo{g}$. \Comment{PRR step}

\State \hskip2em Memoize output $x'$ for $x$.

\State \hskip1em \textbf{else}:

\State \hskip2em Get memoized output $x'$ for $x$.

\State \hskip1em \textbf{end if}

\State  \hskip1em $x''_t \gets \GRR{x';\, \epsIRR}$ over $\oneTo{g}$. \Comment{IRR step}
\State  \hskip1em Send $x''_t$. \Comment{Sanitized data}
\State \textbf{end for}

\end{algorithmic}
\end{algorithm}

\noindent \textbf{Privacy analysis.}
The privacy guarantees of Algorithm~\ref{alg:loloha} are detailed in Theorems~\ref{theo:prr_loloha}, \ref{theo:irr_loloha} and especially~\ref{theo:priv_loloha}.

\begin{theorem}(Single report LDP of memoization) \label{theo:prr_loloha}\\
Let $\calM:V\to\mathscr{H}\times\oneTo{g}$ denote the process of applying the hash and PRR steps of LOLOHA to a single element $v\in V$, producing $\calM(v)=(\rmH, x')$.
Then $\calM$ is $\epsInf$-LDP.
\end{theorem}
\begin{proof}
The parameters for the PRR step are $p=\frac{e^\epsInf}{e^\epsInf+g-1}$ and $q=\frac{1}{e^\epsInf+g-1}$.
For any two possible input values $v_1, v_2 \in V$ and any reported output $(\rmH, x')$, we have
\begin{equation*}
    \frac{\Pr \left[ (\rmH, x') | v_1 \right]}{\Pr \left[ (\rmH, x') | v_2 \right]} \leq \frac{p}{q} = \frac{\frac{e^{\epsInf}}{e^{\epsInf}+g-1}}{\frac{1}{e^{\epsInf}+g-1}} = e^{\epsInf} \textrm{.}
\end{equation*}

\end{proof}

\begin{theorem} \label{theo:irr_loloha} (Single report LDP of LOLOHA)\\
Let $\calM:V\to\mathscr{H}\times\oneTo{g}$ denote the process of applying the hash, PRR, and IRR steps of LOLOHA to a single element $v\in V$, producing $\calM(v)=(\rmH, x'')$.
Then $\calM$ is $\epsOne$-LDP.
\end{theorem}

\begin{proof}
Let $(p_1, q_1)$ denote the parameters for the PRR step and $(p_2, q_2)$, the parameters for the IRR step.
That is, $p_1 = \frac{e^{\epsInf}}{e^{\epsInf}+g-1}$, $q_1 = \frac{1}{e^{\epsInf}+g-1}$, $p_2 = \frac{e^{\epsIRR}}{e^{\epsIRR}+g-1}$, and $q_2 = \frac{1}{e^{\epsIRR}+g-1}$.
If $x''\neq \rmH(v)$, it must have changed during either the PRR or the IRR step, and if $x'' = \rmH(v)$, either it was not changed during either step or it was changed during both.
From this analysis, it can be concluded that for each $y\in\oneTo{g}$, we have
\begin{equation*}
    \Pr[x''=y]
    = \begin{cases}
         p_1 p_2 + q_1 q_2, \text{ if } y = \rmH(v) \textrm{,}\\
         p_1 q_2 + q_1 p_2, \text{ if } y \neq \rmH(v)\textrm{.}
    \end{cases}
\end{equation*}

Therefore, for any two possible input values $v_1, v_2 \in V$ and any output $(\rmH, x'')$, we have,

\begin{equation*}
    \frac{\Pr \left[ (\rmH, x'') | v_1 \right]}{\Pr \left[ (\rmH, x'') | v_2 \right]}
    \leq \frac{p_1 p_2 + q_1 q_2}{p_1 q_2 + q_1 p_2}
    = \frac{e^\epsInf\cdot e^\epsIRR+1\cdot 1}{e^\epsInf\cdot 1+1\cdot e^\epsIRR}.
\end{equation*}

Moreover, since $e^\epsIRR = \frac{e^{\epsInf+\epsOne}-1}{e^\epsInf - e^\epsOne}$, then $e^\epsIRR e^\epsInf + 1 = e^\epsOne (e^\epsIRR + e^\epsInf)$.
Hence,\[
    \frac{\Pr \left[ (\rmH, x'') | v_1 \right]}{\Pr \left[ (\rmH, x'') | v_2 \right]}
    \leq e^\epsOne.
\]
\end{proof}

\begin{theorem} (Privacy protection as $\tau\to\infty$) \label{theo:priv_loloha}\\
The client-side of LOLOHA is $g\epsInf$-\valueLDP{}.
\end{theorem}
\begin{proof}
The non-vectorized memoization mechanism (hash and PRR steps) of LOLOHA is a function $\calM:V\to\oneTo{g}$ that can memorize at most $g$ reports.
For each separate individual report, we known that $\calM$ satisfies $\epsInf$-LDP (Theorem~\ref{theo:prr_loloha}).
Therefore, by sequential composition of at most $g$ results (Proposition~\ref{prop:seq_comp}), $\calM$ satisfies $g\epsInf$-LDP, and LOLOHA satisfies $g\epsInf$-\valueLDP.
\end{proof}

The privacy guarantees of the IRR step (Theorem~\ref{theo:irr_loloha}) degrade according to the number of reports $t \in \oneTo{\tau}$~\cite{rappor,erlingsson2020encode}.
If we let $\epsilon_t$ be the privacy guarantee on the users' values of Algorithm~\ref{alg:loloha} for a fixed user using the data in times $\oneTo{t}$, so that $t=1$ matches exactly $\epsOne$ (Theorem~\ref{theo:irr_loloha}), then we have \(
    \epsOne \leq \epsilon_2 \leq \cdots \leq \epsilon_{\tau} \leq g\epsInf.
\)

Besides, from Algorithm~\ref{alg:loloha}, one can remark that instead of leaking a new $\epsInf$ for each $v \in V$, LOLOHA will only leak $\epsInf$ for each hashed value $x \in \oneTo{g}$.
Therefore, unlike RAPPOR that has a worst-case guarantee of $k\epsInf$-\valueLDP{}, the overall privacy guarantee of our LOLOHA solution will grow proportionally to the new domain size $2 \leq g \ll k$, with worst-case longitudinal privacy of $g\epsInf$-\valueLDP{}.

\subsection{Server-Side of LOLOHA}  \label{sub:server_side}

The server-side algorithm of LOLOHA is described in Algorithm~\ref{alg:loloha_server}, which takes the reported values by $n$ users and aggregates them to estimate the frequencies of each $v\in V$ at each point in time.

\begin{algorithm}[!ht]
\caption{Server-Side of LOLOHA.}
\label{alg:loloha_server}
\begin{algorithmic}[1]
\Statex \textbf{Input:} Constants $0<\epsOne<\epsInf$, and for each user $u\in U$, a hash function $H_u:V\to\oneTo{g}$ and a sequence of hash values $[x''^{(u)}_1,\ldots,x''^{(u)}_\tau]$.
\Statex \textbf{Output:} Matrix with estimations $\hat{f}_\text{LOLOHA}(v)_t$ for each $v\in V$ at each $t \in \oneTo{\tau}$.
\raisebox{-0.8em}{\,}

\State Compute parameters:
\raisebox{-0.4em}{\,}

\Statex\hskip1em
$\epsIRR \gets \ln\left(\frac{e^{\epsInf+\epsOne}-1}{e^\epsInf - e^\epsOne}\right)\quad;\quad n \gets |U|$
\raisebox{-0.9em}{\,}

\Statex\hskip1em
$p_1 \gets \frac{e^{\epsInf}}{e^{\epsInf}+g-1} \quad;\quad q_1' \gets \frac{1}{g}$
\raisebox{-0.9em}{\,}

\Statex\hskip1em
$p_2 \gets \frac{e^{\epsIRR}}{e^{\epsIRR}+g-1} \quad;\quad q_2 \gets \frac{1}{e^{\epsIRR}+g-1}$
\raisebox{-0.9em}{\,}

\State \textbf{for} each time $t \in \oneTo{\tau}$ \textbf{do}:
\State\hskip1em \textbf{for} each $v \in V$ \textbf{do}:
\raisebox{-0.1em}{\,}
\State\hskip2em $C(v) \gets |\{u \in U \mid \rmH_u(v) = x''^{(u)}_t \}|$
\raisebox{-0.6em}{\,}
\State\hskip2em $\hat{f}_L(v)_t \gets \frac{C(v) - n q_1'(p_2-q_2) - nq_2}{n(p_1-q_1')(p_2-q_2)}$ \Comment Eq.~\eqref{eq:est_longitudinal} with $q_1'$.
\raisebox{-0.6em}{\,}
\State\hskip1em \textbf{end for}
\State \textbf{end for}
\State \textbf{return} matrix $[\hat{f}_L(v)_t]_{t, v}$
\end{algorithmic}
\end{algorithm}

For large $n$, the estimations of Algorithm~\ref{alg:loloha_server} are guaranteed to be close to the true population parameters with high probability as explained in Proposition~\ref{prop:estimator_bounds}.
Moreover, one can also compute the LOLOHA approximate variance $\mathbb{V}^*[\hat{f}_{\textrm{LOLOHA}}(v)]$ by replacing the server parameters in Algorithm~\ref{alg:loloha_server} into Eq.~\eqref{var:aprox_longitudinal}.

\begingroup 
\def\v{{v}}
\def\hatf{{\hat f(\v)}}
\def\truef{{f(\v)}}
\def\abserr{{|\hatf-\truef|}}
\def\K{k}
\def\eachv{{v \in V}}

\begin{proposition}\label{prop:estimator_bounds} (Asymptotic utility guarantee of LOLOHA)\\
Fix any arbitrary $t\in\oneTo{\tau}$.
For each $\eachv$, let $\truef$ be the true population probability of producing the value $v$ at time $t$, and let $\hatf \in [0, 1]$ be the estimation produced by Algorithm~\ref{alg:loloha_server} for time $t$.
For any $\beta \in (0,1)$, it holds with probability at least $1-\beta$ that:
\begin{equation*}
    \max_{\eachv}\abserr < \sqrt{\frac{\K}{4 n \beta (p_1-q_1')(p_2-q_2)}} \textrm{.}
\end{equation*}
\end{proposition}

\begin{proof}
Fix $\eachv$, and let $\Delta$ be the random variable given by $\Delta := \hatf - \truef \in [-1,1]$ be a random variable.
Since $\hatf$ is unbiased, we have $\mathbb{E}[\Delta]=0$ and $\mathbb{V}[\Delta] = \mathbb{V}[\hatf]$.
We remark that for any $\delta, \beta \in (0,1)$, among all random variables $\Delta'$ defined in $[-1, 1]$ such that $\mathbb{E}[\Delta']$ and $\Pr[|\Delta'|\geq \delta]=\beta$, the one with minimal variance is the random variable $\Delta^*$ that concentrates a mass of $1-\beta$ at $\Delta'=0$ and two masses of $\beta/2$ at $-\delta$ and $\delta$.
This random variable has variance $\mathbb{V}[\Delta^*] = \beta\delta^2$.
Hence, for arbitrary $\delta\in (0,1)$ and letting in particular $\beta := \Pr[\abserr\geq\delta]$, we conclude that $\mathbb{V}[\hatf] = \mathbb{V}[\abserr] \geq \mathbb{V}[\Delta^*] = \Pr[\abserr\geq\delta] \cdot \delta^2$.
In other words, \[
\Pr[\abserr\geq\delta] \leq \mathbb{V}[\hatf] / \delta^2.
\]
Now, considering all $\eachv$ simultaneously, we obtain $\Pr[\max_{\eachv}\abserr\geq\delta] \leq \sum_\eachv \Pr[\abserr\geq\delta]=(\nicefrac{1}{\delta^2}) \sum_{v\in V} \mathbb{V}[\hatf]$.
By rewriting this equation in terms of confidence, we conclude that with probability at least $1-\beta$,\[
    \max_{\eachv}\abserr < \sqrt{\nicefrac{\sum_{v\in V} \mathbb{V}[\hatf]}{\beta}}.
\]
Lastly, from Eq.~\eqref{var:longitudinal} it can be concluded that $\mathbb{V}[\hatf] \leq \nicefrac{1}{4n(p_1-q_1')(p_2-q_2)}$ because the product $\gamma (1-\gamma)$ is maximal at $\gamma=1/2$.
As a consequence, \(
    \max_{\eachv}\abserr < \sqrt{\nicefrac{\K}{4 n \beta (p_1-q_1')(p_2-q_2)}}.
\)
\end{proof}
\endgroup

\def\vStarRappor{- \frac{\left(- \sqrt{\left(4 e^{\frac{7 \epsilon_{\infty}}{2}} - 4 e^{\frac{5 \epsilon_{\infty}}{2}} - 4 e^{\frac{3 \epsilon_{\infty}}{2}} + 4 e^{\frac{\epsilon_{\infty}}{2}} + e^{4 \epsilon_{\infty}} + 4 e^{3 \epsilon_{\infty}} - 10 e^{2 \epsilon_{\infty}} + 4 e^{\epsilon_{\infty}} + 1\right) e^{\epsilon_{1}}} \left(e^{\epsilon_{1}} - 1\right) \left(e^{\epsilon_{\infty}} - 1\right)^{2} - \left(- \sqrt{\left(4 e^{\frac{7 \epsilon_{\infty}}{2}} - 4 e^{\frac{5 \epsilon_{\infty}}{2}} - 4 e^{\frac{3 \epsilon_{\infty}}{2}} + 4 e^{\frac{\epsilon_{\infty}}{2}} + e^{4 \epsilon_{\infty}} + 4 e^{3 \epsilon_{\infty}} - 10 e^{2 \epsilon_{\infty}} + 4 e^{\epsilon_{\infty}} + 1\right) e^{\epsilon_{1}}} \left(e^{\epsilon_{1}} - 1\right) \left(e^{\epsilon_{\infty}} - 1\right)^{2} - \left(e^{\epsilon_{1}} - 1\right) \left(e^{\epsilon_{\infty}} - 1\right)^{2} \left(e^{\epsilon_{1}} - e^{2 \epsilon_{\infty}} + 2 e^{\epsilon_{\infty}} - 2 e^{\epsilon_{1} + \epsilon_{\infty}} + e^{\epsilon_{1} + 2 \epsilon_{\infty}} - 1\right) + \left(e^{\epsilon_{1}} - e^{2 \epsilon_{\infty}} + 2 e^{\epsilon_{\infty}} - 2 e^{\epsilon_{1} + \epsilon_{\infty}} + e^{\epsilon_{1} + 2 \epsilon_{\infty}} - 1\right) \left(e^{\frac{3 \epsilon_{\infty}}{2}} - e^{\frac{\epsilon_{\infty}}{2}} + e^{\epsilon_{\infty}} - e^{\epsilon_{1} + \frac{\epsilon_{\infty}}{2}} - e^{\epsilon_{1} + \epsilon_{\infty}} + e^{\epsilon_{1} + \frac{3 \epsilon_{\infty}}{2}} + e^{\epsilon_{1} + 2 \epsilon_{\infty}} - 1\right)\right) e^{\frac{\epsilon_{\infty}}{2}} + \left(e^{\epsilon_{1}} - e^{2 \epsilon_{\infty}} + 2 e^{\epsilon_{\infty}} - 2 e^{\epsilon_{1} + \epsilon_{\infty}} + e^{\epsilon_{1} + 2 \epsilon_{\infty}} - 1\right) \left(e^{\frac{3 \epsilon_{\infty}}{2}} - e^{\frac{\epsilon_{\infty}}{2}} + e^{\epsilon_{\infty}} - e^{\epsilon_{1} + \frac{\epsilon_{\infty}}{2}} - e^{\epsilon_{1} + \epsilon_{\infty}} + e^{\epsilon_{1} + \frac{3 \epsilon_{\infty}}{2}} + e^{\epsilon_{1} + 2 \epsilon_{\infty}} - 1\right)\right) \left(- \sqrt{\left(4 e^{\frac{7 \epsilon_{\infty}}{2}} - 4 e^{\frac{5 \epsilon_{\infty}}{2}} - 4 e^{\frac{3 \epsilon_{\infty}}{2}} + 4 e^{\frac{\epsilon_{\infty}}{2}} + e^{4 \epsilon_{\infty}} + 4 e^{3 \epsilon_{\infty}} - 10 e^{2 \epsilon_{\infty}} + 4 e^{\epsilon_{\infty}} + 1\right) e^{\epsilon_{1}}} \left(e^{\epsilon_{1}} - 1\right) \left(e^{\epsilon_{\infty}} - 1\right)^{2} - \left(e^{\epsilon_{1}} - 1\right) \left(e^{\frac{\epsilon_{\infty}}{2}} + 1\right) \left(e^{\epsilon_{\infty}} - 1\right)^{2} \left(e^{\epsilon_{1}} - e^{2 \epsilon_{\infty}} + 2 e^{\epsilon_{\infty}} - 2 e^{\epsilon_{1} + \epsilon_{\infty}} + e^{\epsilon_{1} + 2 \epsilon_{\infty}} - 1\right) - \left(- \sqrt{\left(4 e^{\frac{7 \epsilon_{\infty}}{2}} - 4 e^{\frac{5 \epsilon_{\infty}}{2}} - 4 e^{\frac{3 \epsilon_{\infty}}{2}} + 4 e^{\frac{\epsilon_{\infty}}{2}} + e^{4 \epsilon_{\infty}} + 4 e^{3 \epsilon_{\infty}} - 10 e^{2 \epsilon_{\infty}} + 4 e^{\epsilon_{\infty}} + 1\right) e^{\epsilon_{1}}} \left(e^{\epsilon_{1}} - 1\right) \left(e^{\epsilon_{\infty}} - 1\right)^{2} - \left(e^{\epsilon_{1}} - 1\right) \left(e^{\epsilon_{\infty}} - 1\right)^{2} \left(e^{\epsilon_{1}} - e^{2 \epsilon_{\infty}} + 2 e^{\epsilon_{\infty}} - 2 e^{\epsilon_{1} + \epsilon_{\infty}} + e^{\epsilon_{1} + 2 \epsilon_{\infty}} - 1\right) + \left(e^{\epsilon_{1}} - e^{2 \epsilon_{\infty}} + 2 e^{\epsilon_{\infty}} - 2 e^{\epsilon_{1} + \epsilon_{\infty}} + e^{\epsilon_{1} + 2 \epsilon_{\infty}} - 1\right) \left(e^{\frac{3 \epsilon_{\infty}}{2}} - e^{\frac{\epsilon_{\infty}}{2}} + e^{\epsilon_{\infty}} - e^{\epsilon_{1} + \frac{\epsilon_{\infty}}{2}} - e^{\epsilon_{1} + \epsilon_{\infty}} + e^{\epsilon_{1} + \frac{3 \epsilon_{\infty}}{2}} + e^{\epsilon_{1} + 2 \epsilon_{\infty}} - 1\right)\right) e^{\frac{\epsilon_{\infty}}{2}} + \left(e^{\epsilon_{1}} - e^{2 \epsilon_{\infty}} + 2 e^{\epsilon_{\infty}} - 2 e^{\epsilon_{1} + \epsilon_{\infty}} + e^{\epsilon_{1} + 2 \epsilon_{\infty}} - 1\right) \left(e^{\frac{3 \epsilon_{\infty}}{2}} - e^{\frac{\epsilon_{\infty}}{2}} + e^{\epsilon_{\infty}} - e^{\epsilon_{1} + \frac{\epsilon_{\infty}}{2}} - e^{\epsilon_{1} + \epsilon_{\infty}} + e^{\epsilon_{1} + \frac{3 \epsilon_{\infty}}{2}} + e^{\epsilon_{1} + 2 \epsilon_{\infty}} - 1\right)\right)}{n \left(e^{\frac{\epsilon_{\infty}}{2}} - 1\right)^{2} \left(- 2 \sqrt{\left(4 e^{\frac{7 \epsilon_{\infty}}{2}} - 4 e^{\frac{5 \epsilon_{\infty}}{2}} - 4 e^{\frac{3 \epsilon_{\infty}}{2}} + 4 e^{\frac{\epsilon_{\infty}}{2}} + e^{4 \epsilon_{\infty}} + 4 e^{3 \epsilon_{\infty}} - 10 e^{2 \epsilon_{\infty}} + 4 e^{\epsilon_{\infty}} + 1\right) e^{\epsilon_{1}}} \left(e^{\epsilon_{1}} - 1\right) \left(e^{\epsilon_{\infty}} - 1\right)^{2} - \left(e^{\epsilon_{1}} - 1\right) \left(e^{\epsilon_{\infty}} - 1\right)^{2} \left(e^{\epsilon_{1}} - e^{2 \epsilon_{\infty}} + 2 e^{\epsilon_{\infty}} - 2 e^{\epsilon_{1} + \epsilon_{\infty}} + e^{\epsilon_{1} + 2 \epsilon_{\infty}} - 1\right) + 2 \left(e^{\epsilon_{1}} - e^{2 \epsilon_{\infty}} + 2 e^{\epsilon_{\infty}} - 2 e^{\epsilon_{1} + \epsilon_{\infty}} + e^{\epsilon_{1} + 2 \epsilon_{\infty}} - 1\right) \left(e^{\frac{3 \epsilon_{\infty}}{2}} - e^{\frac{\epsilon_{\infty}}{2}} + e^{\epsilon_{\infty}} - e^{\epsilon_{1} + \frac{\epsilon_{\infty}}{2}} - e^{\epsilon_{1} + \epsilon_{\infty}} + e^{\epsilon_{1} + \frac{3 \epsilon_{\infty}}{2}} + e^{\epsilon_{1} + 2 \epsilon_{\infty}} - 1\right)\right)^{2}}}
\def\vStarOptLonUE{\frac{4 e^{\epsilon_{1}}}{n \left(e^{2 \epsilon_{1}} - 2 e^{\epsilon_{1}} + 1\right)}}
\def\vStarLonGRR{- \frac{\left(\left(1 - e^{\epsilon_{1} + \epsilon_{\infty}}\right) \left(- \left(k - 1\right) \left(e^{\epsilon_{1}} - e^{\epsilon_{\infty}}\right) + e^{\epsilon_{1} + \epsilon_{\infty}} - 1\right) + \left(e^{\epsilon_{1}} - e^{\epsilon_{\infty}}\right) \left(- k - e^{\epsilon_{\infty}} + 2\right) \left(\left(k - 1\right) \left(e^{\epsilon_{1}} - e^{\epsilon_{\infty}}\right) - e^{\epsilon_{1} + \epsilon_{\infty}} + 1\right)\right) \left(- \left(1 - e^{\epsilon_{1} + \epsilon_{\infty}}\right) \left(\left(k - 1\right) \left(- e^{\epsilon_{1}} + e^{\epsilon_{\infty}}\right) + e^{\epsilon_{1} + \epsilon_{\infty}} - 1\right) - \left(- e^{\epsilon_{1}} + e^{\epsilon_{\infty}}\right) \left(- k - e^{\epsilon_{\infty}} + 2\right) \left(- \left(k - 1\right) \left(e^{\epsilon_{1}} - e^{\epsilon_{\infty}}\right) + e^{\epsilon_{1} + \epsilon_{\infty}} - 1\right) - \left(k + e^{\epsilon_{\infty}} - 1\right) \left(\left(k - 1\right) \left(- e^{\epsilon_{1}} + e^{\epsilon_{\infty}}\right) + e^{\epsilon_{1} + \epsilon_{\infty}} - 1\right) \left(- \left(k - 1\right) \left(e^{\epsilon_{1}} - e^{\epsilon_{\infty}}\right) + e^{\epsilon_{1} + \epsilon_{\infty}} - 1\right)\right) \left(\left(k - 1\right) \left(e^{\epsilon_{1}} - e^{\epsilon_{\infty}}\right) - e^{\epsilon_{1} + \epsilon_{\infty}} + 1\right)}{n \left(1 - e^{\epsilon_{\infty}}\right)^{2} \left(\left(1 - e^{\epsilon_{1} + \epsilon_{\infty}}\right) \left(- \left(k - 1\right) \left(e^{\epsilon_{1}} - e^{\epsilon_{\infty}}\right) + e^{\epsilon_{1} + \epsilon_{\infty}} - 1\right) + \left(e^{\epsilon_{1}} - e^{\epsilon_{\infty}}\right) \left(\left(k - 1\right) \left(e^{\epsilon_{1}} - e^{\epsilon_{\infty}}\right) - e^{\epsilon_{1} + \epsilon_{\infty}} + 1\right)\right)^{2} \left(\left(k - 1\right) \left(- e^{\epsilon_{1}} + e^{\epsilon_{\infty}}\right) + e^{\epsilon_{1} + \epsilon_{\infty}} - 1\right)}}
\def\vStarDBitFlip{\frac{b}{2 d n \sinh{\left(\frac{\epsilon_{\infty}}{2} \right)}}}
\def\vStarLoloha{- \frac{\left(\left(e^{\epsilon_{1}} - e^{\epsilon_{\infty}}\right) \left(g + e^{\epsilon_{\infty}} - 2\right) - e^{\epsilon_{1} + \epsilon_{\infty}} + 1\right) \left(\left(e^{\epsilon_{1}} - e^{\epsilon_{\infty}}\right) \left(g + e^{\epsilon_{\infty}} - 2\right) - \left(g + e^{\epsilon_{\infty}} - 1\right) \left(\left(g - 1\right) \left(e^{\epsilon_{1}} - e^{\epsilon_{\infty}}\right) - e^{\epsilon_{1} + \epsilon_{\infty}} + 1\right) - e^{\epsilon_{1} + \epsilon_{\infty}} + 1\right)}{n \left(e^{\epsilon_{\infty}} - 1\right)^{2} \left(e^{\epsilon_{1}} - e^{\epsilon_{\infty}} + e^{\epsilon_{1} + \epsilon_{\infty}} - 1\right)^{2}}}

\subsection{Selecting and Optimizing Parameter $g$} \label{sub:selecting_g}

\noindent \textbf{Binary LOLOHA (BiLOLOHA).} Following Theorem~\ref{theo:priv_loloha}, the strongest longitudinal privacy protection of LOLOHA is when $g=2$.

\noindent \textbf{Optimal LOLOHA (OLOLOHA).} To maximize the utility of LOLOHA, we find the optimal $g$ value by taking the partial derivative of $\mathbb{V}^*[\hat{f}_{\textrm{LOLOHA}}(v)]$ with respect to $g$.
Let $\epsOne = \alpha \epsInf$, for $\alpha \in (0,1)$.
This partial derivative is a function in terms of $\epsInf$ and $\alpha$, or alternatively, in terms of $a=e^{\epsilon_\infty}$ and $b=e^{\alpha \epsilon_\infty}$, and it is minimized when $g$ equals (\cf{} development in repository~\cite{artifact_loloha}):

\begin{equation} \label{eq:opt_g}
    g=1{+}{\max}\left(1,\left \lfloor \frac{1 {-} a^{2} {+} \sqrt{a^{4} {-} 14 a^{2} {+} 12 a b (1{-} a b){+} 12 a^3 b {+} 1} }{6 (a {-} b)} \right \rceil \right ) \textrm{,}
\end{equation}

\noindent in which $\lfloor . \rceil$ means rounding to the closest integer.
Fig.~\ref{fig:analysis_g} illustrates the optimal $g$ selection with Eq.~\eqref{eq:opt_g} by varying the longitudinal privacy guarantee $\epsInf=[0.5, 1, \ldots, 4.5, 5]$ and $\alpha \in \{0.1, 0.2, \ldots, 0.6\}$. 
From Fig.~\ref{fig:analysis_g}, one can remark that in high privacy regimes (\ie, low $\epsilon$ values), the optimal $g$ is binary (\ie, our BiLOLOHA protocol with $g=2$). 
As $\epsInf$ or/and $\epsOne=\alpha \epsInf$ get(s) higher (low privacy regimes), the optimal $g$ is non-binary, which can maximize utility with a cost in the overall longitudinal privacy $g\epsInf$-\valueLDP{}, for $g > 2$.

\begin{figure}[t]
    \centering
    \includegraphics[width=0.8\linewidth]{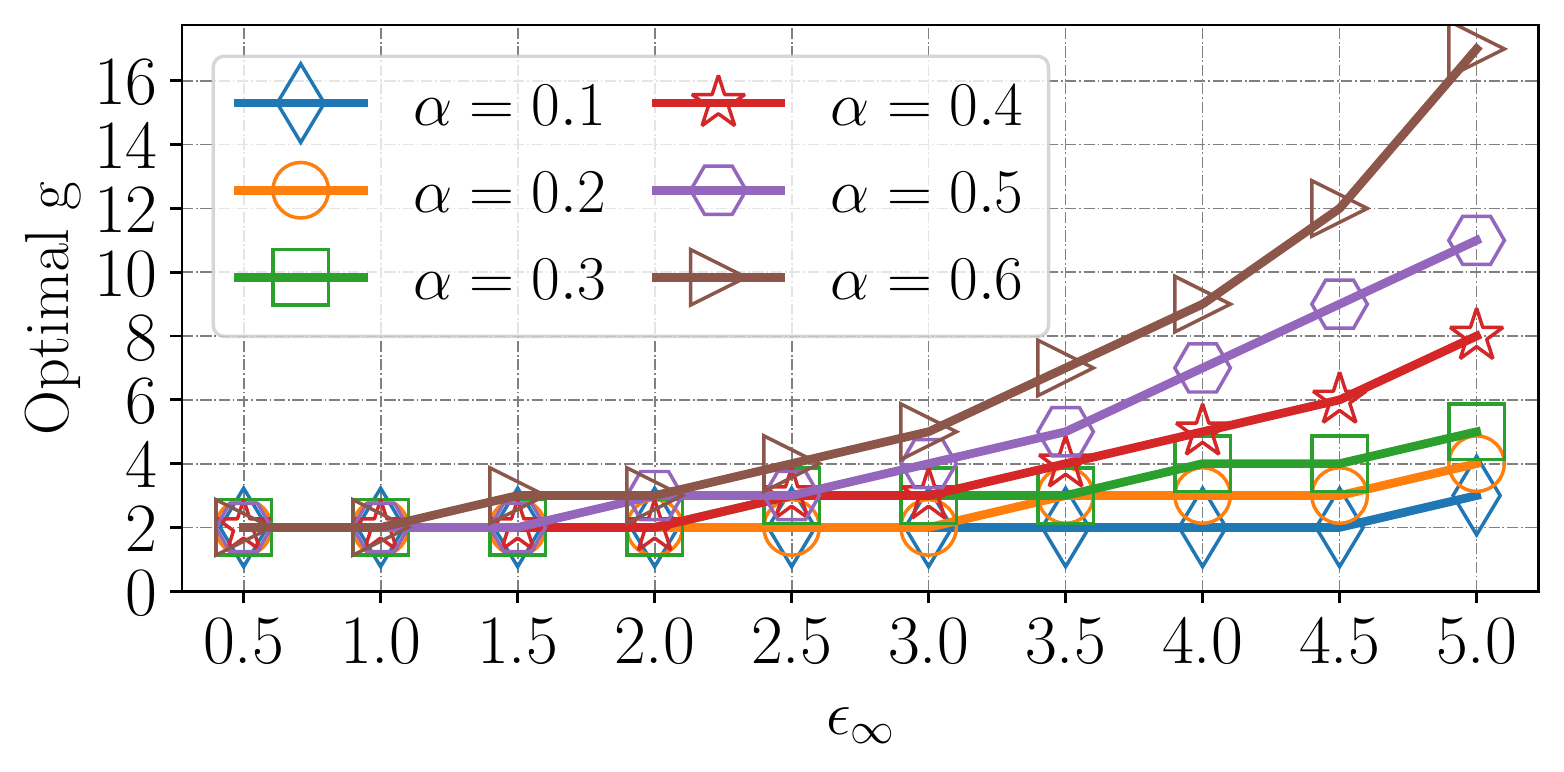}
    \caption{Optimal $g$ selection for our OLOLOHA protocol by varying the longitudinal $\epsInf$ and first report $\epsOne=\alpha \epsInf$ privacy guarantees, for $\alpha \in \{0.1, 0.2, \ldots, 0.6\}$.}
    \label{fig:analysis_g}
\end{figure}

\section{Theoretical Comparison} \label{sec:theoretical_comparison}

In this section, we compare LOLOHA with the state-of-the-art protocols described in the previous Section~\ref{sub:long_fo} from a theoretical point of view.
Table~\ref{table-theory} shows a summary of the main characteristics of these protocols, excluding utility.

For the theoretical utility, numerical analysis is preferred over an analytical one because the formulas of variance and approximate variance are excessively complex.
For L-OSUE and $d$BitFlipPM, the approximate variances are $\vStarOptLonUE{}$ and $\vStarDBitFlip{}$ respectively, but for the other protocols, the formulas are provided only in the repository~\cite{artifact_loloha} since they are excessively verbose for this document.

In order to evaluate numerically the approximate variance $\mathbb{V}^*$ of LOLOHA in comparison with state-of-the-art ones~\cite{rappor,Arcolezi2021_allomfree}, for each protocol, we set the longitudinal privacy guarantee $\epsInf$ (upper bound) and the first report privacy guarantee $\epsOne=\alpha \epsInf$ (lower bound), for $\alpha \in (0,1)$. 
This allows to obtain parameters $p_1, q_1, p_2, q_2$ for each protocol, which are then used to compute their approximate variance with Eq.~\eqref{var:aprox_longitudinal}.

Fig.~\ref{fig:analysis_var} illustrates the numerical values of the approximate variance for our LOLOHA protocols, RAPPOR~\cite{rappor}, and L-OSUE~\cite{Arcolezi2021_allomfree} with $n=10000$, $\epsInf=[0.5, 1, \ldots, 4.5, 5]$, and $\alpha \in \{0.1, 0.2, \ldots, 0.6\}$.
From Fig.~\ref{fig:analysis_var}, one can remark that all protocols have similar variance values when $\alpha \leq 0.3$ with only a small difference when $\epsInf$ is high.
However, in low privacy regimes, \ie, when $\epsInf$ and $\alpha$ are high, BiLOLOHA is the least performing protocol in terms of utility, accompanied by RAPPOR.
Indeed, our OLOLOHA protocol has a very similar utility as the optimized L-OSUE~\cite{Arcolezi2021_allomfree} protocol, which indicates a clear connection also found between their one-round versions~\cite{tianhao2017}, \ie, OLH and OUE.

\begin{figure*}[ht]
    \centering
    \includegraphics[width=0.96\linewidth]{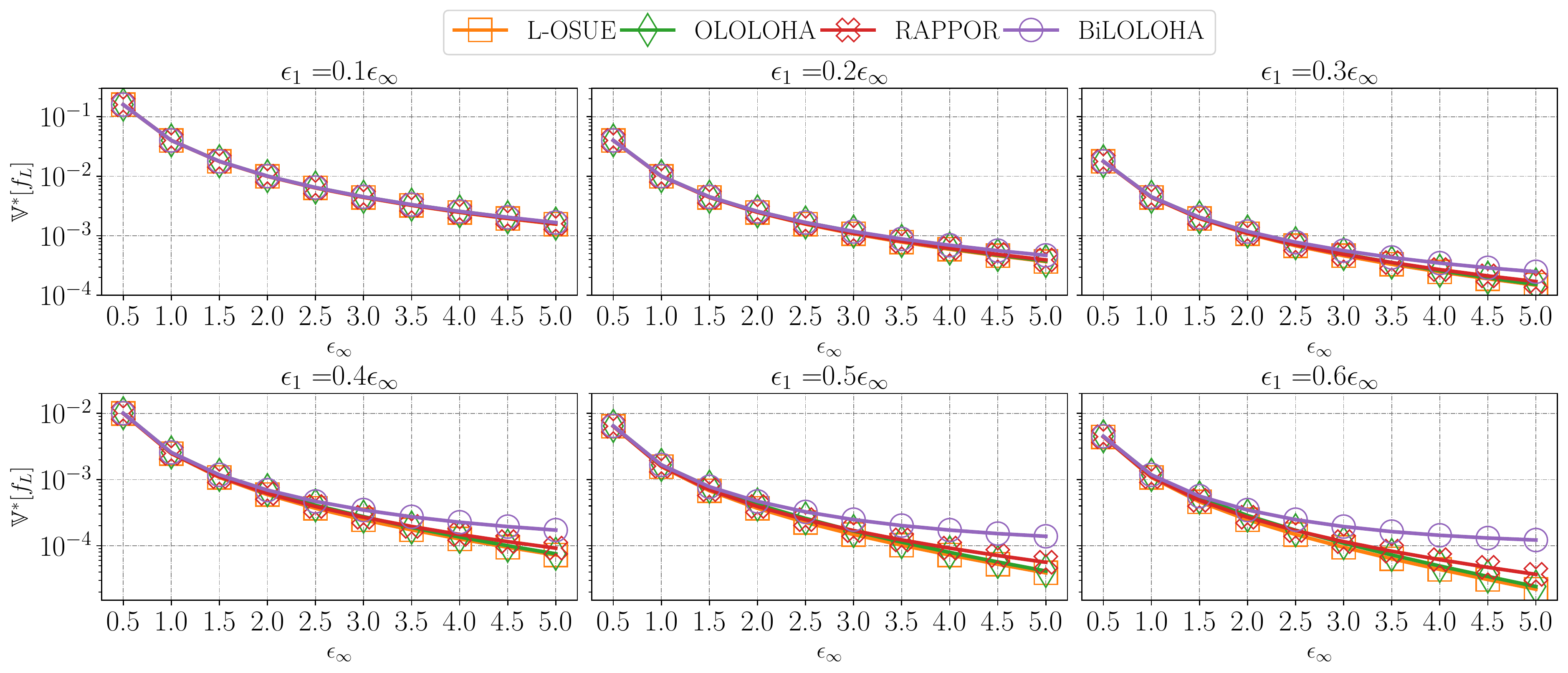}
    \caption{Numerical values of the approximate variance $\mathbb{V}^*[\hat{f}_L(v)]$ in Eq.~\eqref{var:aprox_longitudinal} of our LOLOHA protocols, RAPPOR~\cite{rappor}, and L-OSUE~\cite{Arcolezi2021_allomfree} varying the longitudinal $\epsInf$ and first report $\epsOne=\alpha \epsInf$ privacy guarantees, for $\alpha \in \{0.1, 0.2, \ldots, 0.6\}$.}
    \label{fig:analysis_var}
\end{figure*}

\setlength{\tabcolsep}{4pt}
\begin{table}
\begin{tabular}{|c|c|c|c|}
\hline
\textbf{Protocol }& \textbf{Comm. }& \textbf{Server}& \textbf{Privacy loss}\\
 & bits per user & run-time & budget\\
 & per time step & complexity & consumption\\
\hline
LOLOHA & $\ceil{\log_2 g}$ & $n\,k$ & $g\,\epsInf$\\
\hline
L-GRR~\cite{Arcolezi2021_allomfree} & $\ceil{\log_2 k}$ & $n$ & $k\,\epsInf$ \\
\hline
RAPPOR~\cite{rappor} & $k$ & $n\,k$ & $k\,\epsInf$\\
\hline
L-OSUE~\cite{Arcolezi2021_allomfree} & $k$ & $n\,k$ & $k\,\epsInf$ \\
\hline
$d$BitFlipPM~\cite{microsoft} & $d$ & $n\,b$ & $\min(d+1, b)\,\epsInf$\\
\hline
\end{tabular}
\caption{Theoretical comparison of the protocols.}
\label{table-theory}
\end{table}

Though not included in our analysis, the L-GRR protocol from~\cite{Arcolezi2021_allomfree} has shown to be very sensitive to $k$ (a parameter on which its variance depends on), leading to extremely high values that would obfuscate the curves of the other protocols in Fig.~\ref{fig:analysis_var}. 
However, L-GRR is ideal when $k$ is small, which is the case for instance for binary attributes.
Besides, we also did not numerically compare our protocols with $d$BitFlipPM as it only has a single round of sanitization. 
A proper comparison with $d$BitFlipPM would be only considering the PRR step of our LOLOHA protocols.
Therefore, by comparing the approximate variances of double randomization protocols, we can conclude that our LOLOHA protocols preserve as much utility as state-of-the-art protocols~\cite{rappor,Arcolezi2021_allomfree}. 

Moreover, from Table~\ref{table-theory}, LOLOHA has less communication cost than L-UE and similar server time computation, which is advantageous for large-scale system deployment to monitor frequency longitudinally.
In addition, one clear limitation of RAPPOR, L-OSUE, and L-GRR is that they do not support even small data changes of the user's actual data~\cite{microsoft}, which requires to invoke the whole algorithm again on the new value. 
Therefore, following Definition~\ref{def:ldp_memo} and Proposition~\ref{prop:seq_comp}, the overall privacy guarantee of RAPPOR, L-OSUE, and L-GRR, for all user's true value $v \in V$ (assuming the user's value will change periodically) will grow proportionally to the number of \textit{data changes}, with worst-case longitudinal privacy of $k\epsInf$-\valueLDP.

On the other hand, with $d$BitFlipPM, the overall privacy guarantee of $d$BitFlipPM for all user's true value $v \in V$ (assuming the user's value will change periodically) will grow proportionally to the number \textit{of bits $d$} or the number of \textit{bucket changes}, with worst-case longitudinal privacy of $\min(d+1,b)\epsInf$-\valueLDP{} (\cf{} Definition~\ref{def:ldp_memo} and Proposition~\ref{prop:seq_comp}).
However, there is a loss of information due to both the generalization of the original domain size $k$ to $b$ buckets and due to sampling only $d$ bits.
Besides, the $d$BitFlipPM protocol is vulnerable to detecting high data changes (\ie, change of real bucket) as there is no second round of sanitization (\ie, IRR step)~\cite{Xue2022}. 
This data change detection problem is (to some extent) minimized when $d$ is small.

\section{Experimental Evaluation}  \label{sec:results_discussion}

In this section, we present the setup of our experiments and the experimental results of our LOLOHA protocols in comparison with the state-of-the-art.

\subsection{Setup of Experiments}  \label{sub:setup}

The main goal of our experiments is to study the effectiveness of our proposed LOLOHA protocols on longitudinal frequency estimates through multiple $\tau>1$ data collections. 
In particular, we aim to show that our LOLOHA protocols (i) maintain competitive utility to state-of-the-art memoization-based LDP protocols~\cite{rappor,microsoft,Arcolezi2021_allomfree} while (ii) minimize longitudinal privacy loss.
With these objectives in mind, we run experiments using both synthetic and real-world datasets.

\noindent \textbf{Environment.} All algorithms are implemented in Python 3 with Numpy and Numba libraries. 
The codes we develop for all experiments are available in the repository~\cite{artifact_loloha}. 
Since LDP algorithms are randomized, we report average results over 20 runs.

\noindent \textbf{Datasets.} 
We use the following real and synthetic datasets. 

\begin{itemize}
    
    \item \textbf{Syn.} To simulate the deployment of~\cite{microsoft} to collect data every 6 hours, we generate a synthetic dataset with $k=360$ (\ie., the number of minutes in 6 hours), $n=10000$ users, and $\tau=120$ data collections (\ie, 4x over 30 days). 
    For each user, the value at the first timestamp follows a Uniform distribution. 
    For each subsequent time, a change can occur with probability $p_{ch}=0.25$, with value following a Uniform distribution too.

    \item \textbf{Adult.} This is a classical dataset from the UCI machine learning repository~\cite{uci} with $n=45222$ samples after cleaning. 
    We only selected the ``hours-per-week" attribute with $k=96$. 
    To simulate multiple data collections, we randomly permuted the data $\tau=260$ times (\ie, 52 weeks over 5-years). 
    Note that the real frequency remains the same but each user has a random private sequence.

    \item \textbf{DB\_MT.} This dataset is produced by the \texttt{folktables} Python package~\cite{ding2021retiring} that provides access to datasets derived from the US Census. 
    We selected the survey year 2018 and the ``Montana'' state, which results in $n=10336$ samples. 
    To simulate $\tau=80$ counter data collections, we selected all person record-replicate weights attributes\footnote{\url{https://www.census.gov/programs-surveys/acs/microdata/documentation.html}.}, \ie, PWGTP1, \ldots, PWGTP80. 
    The total number of unique values among all columns is $k=1412$.
    
    \item \textbf{DB\_DE.} Similar to DB\_MT, we selected the ``Delaware'' state, which results in $n=9123$, $\tau=80$, and $k=1234$.

\end{itemize}

\noindent \textbf{Methods evaluated.} We consider for evaluation the following longitudinal LDP protocols: 

\begin{itemize}
    
    \item \textit{RAPPOR}. The utility-oriented protocol from~\cite{rappor} based on SUE (\cf{} Section~\ref{sub:rappor}).
    
    \item \textit{L-OSUE}. The optimized L-UE protocol from~\cite{Arcolezi2021_allomfree} (\cf{} Section~\ref{sub:l_osue}).
    
    \item \textit{L-GRR}. The optimized longitudinal protocol from~\cite{Arcolezi2021_allomfree} when $k$ is small (\cf{} Section~\ref{sub:l_grr}).

    \item \textit{$d$BitFlipPM}. The one-round randomization mechanism from~\cite{microsoft} with $d \in \{1, b\}$, referred respectively as $1$BitFlipPM and $b$BitFlipPM, in which the former $1$BitFlipPM is tuned for privacy and the latter $b$BitFlipPM for utility (\cf{} Section~\ref{sub:microsoft})
    
    \item \textit{LOLOHA}. Our protocols following Algorithm~\ref{alg:loloha}, which are BiLOLOHA with $g=2$ adjusted for privacy and OLOLOHA with $g$ following Eq.~\eqref{eq:opt_g} tuned for utility.

\end{itemize}

\noindent \textbf{Privacy metrics.} We vary the longitudinal privacy parameter in the range $\epsInf=[0.5, 1, \ldots, 4.5, 5]$ and $\epsOne = \alpha \epsInf$, for $\alpha \in \{0.4, 0.5, 0.6\}$, to compare our experimental results with numerical ones from Section~\ref{sec:theoretical_comparison} (with higher visibility).

\noindent \textbf{Performance metrics.} To evaluate our results, we use the MSE averaged by the number of data collection $\tau$, denoted by $MSE_{avg}$. 
Thus, for each time $t \in \oneTo{\tau}$, we compute for each value $v \in V$ the estimated frequency $\hat{f}_L(v)_t$ and the real one $f(v)_t$ and calculate their differences before averaging by $\tau$.
More formally,

\begin{equation} \label{eq:mse_avg}
    MSE_{avg} = \frac{1}{\tau} \sum_{t \in \oneTo{\tau}} \frac{1}{|V|} \sum_{v \in V}\left(f(v)_t - \hat{f}_L(v)_t \right)^2  \textrm{.}
\end{equation}

We also assess the averaged longitudinal privacy loss for all users, denoted by $\check{\epsilon}_{avg}$.
More precisely, after the end of all data collections $\tau$, we compute for each user $u \in U$ their overall longitudinal privacy loss $\check{\epsilon}_{\infty}^{(u)}$ and average by $n$.
For example, RAPPOR (and L-GRR and L-OSUE) leaks a new $\epsInf$ in each data change with $\check{\epsilon}_{\infty} \leq k\epsInf$, while LOLOHA protocols leak a new $\epsInf$ in each hash value change with $\check{\epsilon}_{\infty} \leq g \epsInf$.
More formally,
\begin{equation} \label{eq:eps_avg}
    \check{\epsilon}_{avg} = \frac{1}{n} \sum_{u \in U} \check{\epsilon}_{\infty}^{(u)}  \textrm{.}
\end{equation}

Finally, for the $d$BitFlipPM protocol, we also evaluate the percentage of users in which an attacker can identify \textbf{all} (bucket) data change points (\ie, \textit{worst-case} analysis) due to different PRR reports throughout the $\tau$ data collections. 

\subsection{Results}  \label{sub:results}

First, we compare the utility performance of our LOLOHA protocols with all four state-of-the-art memoization-based protocols for frequency monitoring under LDP guarantees, namely, RAPPOR~\cite{rappor}, L-OSUE~\cite{Arcolezi2021_allomfree}, L-GRR~\cite{Arcolezi2021_allomfree}, and $d$BitFlipPM~\cite{microsoft}, for $d \in \{1, b\}$. 
Fig.~\ref{fig:mse_results} illustrates the $MSE_{avg}$ metric in Eq.~\eqref{eq:mse_avg} for all methods and all Syn, Adult, DB\_MT, and DB\_DE datasets, by varying the longitudinal $\epsInf$ and first report $\epsOne=\alpha \epsInf$ privacy guarantees, for $\alpha \in \{0.4, 0.5, 0.6\}$.
On the one hand, since $k\leq 360$ for Syn and Adult datasets, when implementing $d$BitFlipPM, we select $b=k$ to estimate the same $k$-bins histogram as all other methods in Figs.~\ref{subfig:mse_a} and~\ref{subfig:mse_b}. 
On the other hand, we select $b=\lfloor k/4 \rfloor$ bins for both DB\_MT ($k=1412$) and DB\_DE ($k=1234$) datasets, but we did not include the error metric of $d$BitFlipPM in Figs.~\ref{subfig:mse_c} and~\ref{subfig:mse_d} as the error is five orders of magnitude higher due to histograms of different sizes ($b<k$).

Fig.~\ref{fig:mse_results} shows that the experimental results with all datasets match the numerical results of variance values from Fig.~\ref{fig:analysis_var} for our LOLOHA protocols, RAPPOR, and L-OSUE.
More specifically, our OLOLOHA protocol has similar utility to the optimized L-OSUE protocol, a relationship also find between their one-round versions OLH and OUE in~\cite{tianhao2017}.
In high privacy regimes, all four protocols, \ie, RAPPOR, L-OSUE, BiLOLOHA, and OLOLOHA have very similar utility. 
In low privacy regimes, L-OSUE and OLOLOHA outperforms both RAPPOR and BiLOLOHA. 
The least performing longitudinal LDP protocols are L-GRR and $1$BitFlipPM, the former due to high domain sizes $k$, as shown in~\cite{Arcolezi2021_allomfree}, and the latter due to sampling only a single $d=1$ bit out of $b$ ones.
The $b$BitFlipPM protocol outperforms all experimented longitudinal LDP protocols due to having only a single round of sanitization (\ie, the PRR step) and by reporting all $d=b$ bits, which is consistent with~\cite{microsoft} (the larger $d$ the greater the utility).

However, increasing the number of bits $d$ the users must report negatively impacts privacy, as each new input value has a high probability of generating a new output value, which will be detected by the server.
For instance, for both $d$BitFlipPM protocols, for $d\in\{1,b\}$, Table~\ref{tab:change_detection} exhibits the percentage of users in which \textbf{all} bucket changes were detected by the server due to different PRR responses throughout $\tau$ data collections, for all Syn, Adult, DB\_MT, and DB\_DE datasets.
Remark that when $d=1$, the protocol is adjusted for privacy, thus being less vulnerable with respect to privacy with only a small percentage ($< 1\%$) of users that the server always detected a different randomized output due to different input values.
Besides, one can note that the percentage of attacked users decreases as $\epsInf$ gets higher when $d=1$. 
The intuition is that the probability of randomizing the single bit will be smaller with high $\epsInf$, thus generating the same report many times.
On the other hand, the $b$BitFlipPM protocol is tuned for utility, which increased the probability of \textit{always} generating a new randomized output due to new input values and, thus leading to 100\% of detection for all four datasets.
Though we only perform both extreme cases (lower $d=1$ and upper $d=b$ bounds), one can picture the privacy-utility trade-off of $d$BitFlipPM for other  $d$ values in between our results of Fig.~\ref{fig:mse_results} and Table~\ref{tab:change_detection}.

\begin{table}[t]
\setlength{\tabcolsep}{2.5pt}
\small
\caption{Percentage of users in which the server detected \textbf{all} data change points for $d$BitFlipPM, for $d \in \{1,b\}$, and all Syn, Adult, DB\_MT, and DB\_DE datasets.}
\label{tab:change_detection}
\centering
\begin{tabular}{|c| c c c c| c c c c|}
\hline
\multicolumn{1}{|c|}{\multirow{2}{*}{$\epsInf$}} & \multicolumn{4}{c|}{$d=1$} & \multicolumn{4}{c|}{$d=b$} \\ \cline{2-9}
& Syn & Adult & DB\_MT &DB\_DE & Syn & Adult & DB\_MT &DB\_DE \\
\hline
0.5 &  0\% &  0\%  & 0.0048\%  &  0\%  &  100\% &  100\% &  100\% &  100\% \\
1.0 &  0\% &  0\%  & 0.0044\%  &  0\%  &  100\% &  100\% &  100\% &  100\% \\
1.5 &  0\% &  0\%  & 0.0048\%  &  0\%  &  100\% &  100\% &  100\% &  100\% \\
2.0 &  0\% &  0\%  & 0.0039\%  &  0\%  &  100\% &  100\% &  100\% &  100\% \\
2.5 &  0\% &  0\%  & 0.0024\%  &  0\%  &  100\% &  100\% &  100\% &  100\% \\
3.0 &  0\% &  0\%  & 0.0024\%  &  0\%  &  100\% &  100\% &  100\% &  100\% \\
3.5 &  0\% &  0\%  & 0.0024\%  &  0\%  &  100\% &  100\% &  100\% &  100\% \\
4.0 &  0\% &  0\%  & 0.0019\%  &  0\%  &  100\% &  100\% &  100\% &  100\% \\
4.5 &  0\% &  0\%  & 0.0010\%  &  0\%  &  100\% &  100\% &  100\% &  100\% \\
5.0 &  0\% &  0\%  & 0.0010\%  &  0\%  &  100\% &  99.99\% &  100\% &  100\% \\
\hline
\end{tabular}
\end{table}

\begin{figure*}[!ht]
\begin{subfigure}{2\columnwidth}
  \centering
  \includegraphics[width=0.96\linewidth]{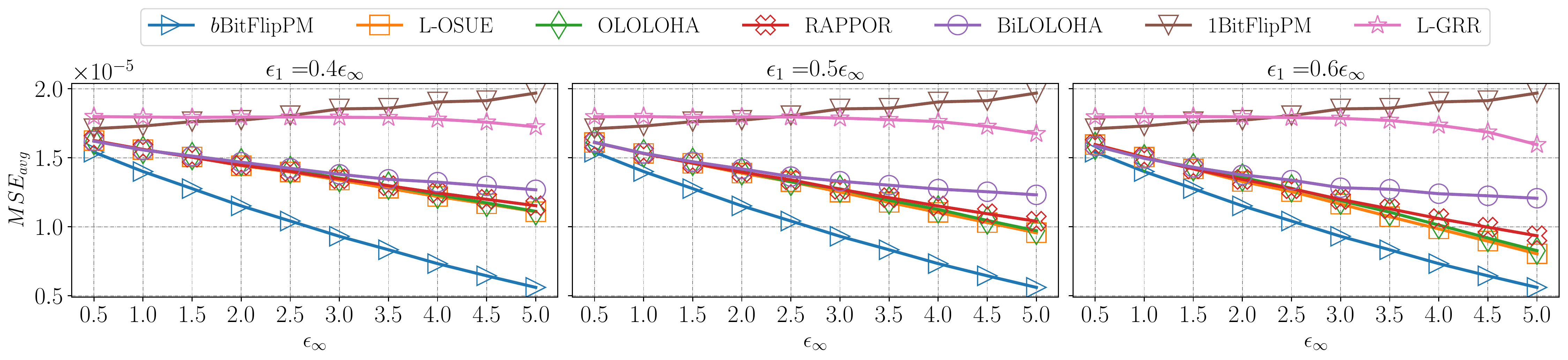}
  \caption{Syn dataset: $k=360$, $n=10000$, and $\tau=120$.}
  \label{subfig:mse_a}
\end{subfigure}%
\\
\begin{subfigure}{2\columnwidth}
  \centering
  \includegraphics[width=0.96\linewidth]{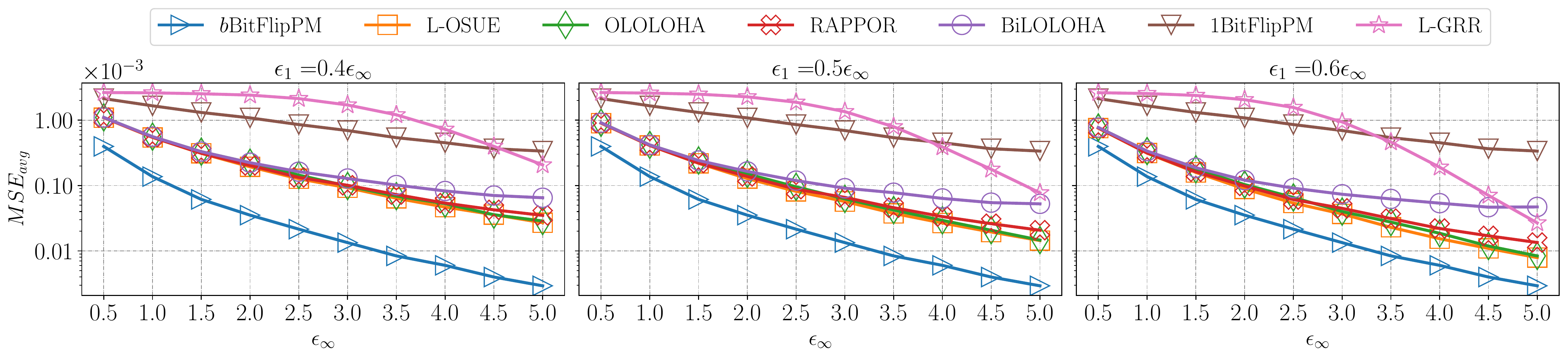}
  \caption{Adult dataset: $k=96$, $n=45222$, and $\tau=260$.}
  \label{subfig:mse_b}
\end{subfigure}
\\
\begin{subfigure}{2\columnwidth}
  \centering
  \includegraphics[width=0.96\linewidth]{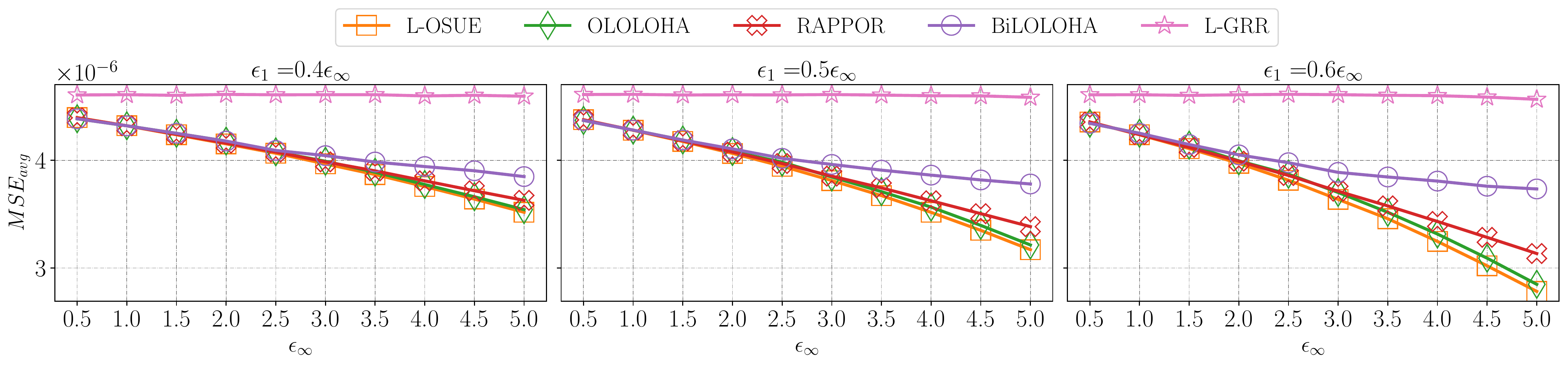}
  \caption{DB\_MT dataset: $k=1412$, $n=10336$, and $\tau=80$.}
  \label{subfig:mse_c}
\end{subfigure}%
\\
\begin{subfigure}{2\columnwidth}
  \centering
  \includegraphics[width=0.96\linewidth]{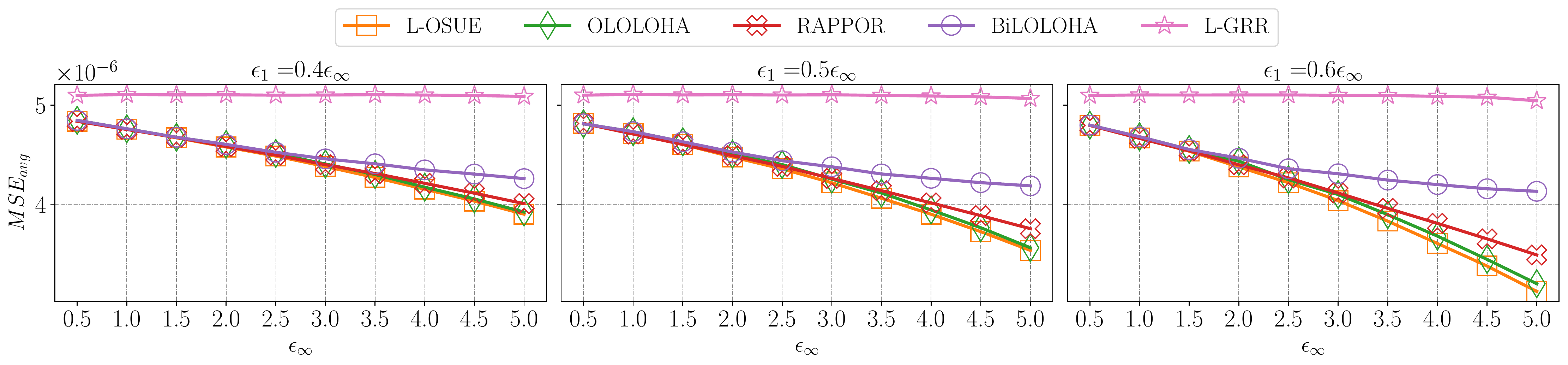}
  \caption{DB\_DE dataset: $k=1234$, $n=9123$, and $\tau=80$.}
  \label{subfig:mse_d}
\end{subfigure}
\caption{Averaged MSE for $\tau$ data collections in Eq.~\eqref{eq:mse_avg} by varying the longitudinal $\epsInf$ and first report $\epsOne=\alpha \epsInf$ privacy guarantees, for $\alpha \in \{0.4, 0.5, 0.6\}$, on (a) Syn, (b) Adult, (c) DB\_MT, and (d) DB\_DE datasets. The evaluated methods are: $d$BitFlipPM~\cite{microsoft}, L-OSUE~\cite{Arcolezi2021_allomfree}, RAPPOR~\cite{rappor}, L-GRR~\cite{Arcolezi2021_allomfree}, and our LOLOHA protocols.}
\label{fig:mse_results}
\end{figure*}

We now analyze the longitudinal privacy guarantees of our LOLOHA protocols in comparison with the state-of-the-art memoization-based LDP protocols. 
Fig.~\ref{fig:eps_results} illustrates the $\check{\epsilon}_{avg}$ metric in Eq.~\eqref{eq:eps_avg} for all methods and all Syn, Adult, DB\_MT, and DB\_DE datasets, by varying the longitudinal $\epsInf$ and first report $\epsOne=\alpha \epsInf$ privacy guarantees, for $\alpha \in \{0.4, 0.5, 0.6\}$.
Notice that the results of $d$BitFlipPM protocols in Figs.~\ref{subfig:eps_a} and~\ref{subfig:eps_b} are with $b=k$ buckets and in Figs.~\ref{subfig:eps_c} and~\ref{subfig:eps_d} are with $b=\lfloor k/4 \rfloor$ buckets.

\begin{figure*}[!ht]
\begin{subfigure}{2\columnwidth}
  \centering
  \includegraphics[width=0.96\linewidth]{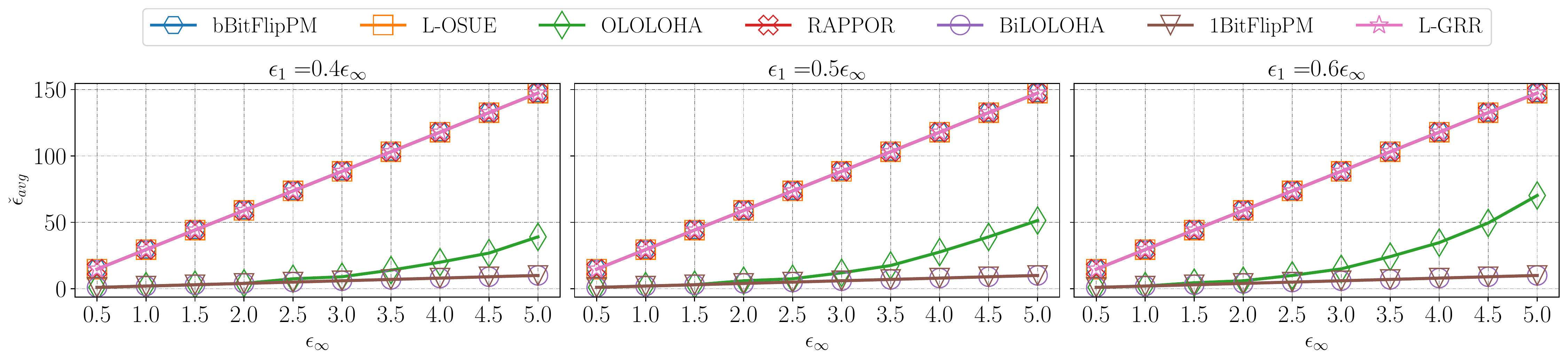}
  \caption{Syn dataset: $k=360$, $n=10000$, and $\tau=120$.}
  \label{subfig:eps_a}
\end{subfigure}%
\\
\begin{subfigure}{2\columnwidth}
  \centering
  \includegraphics[width=0.96\linewidth]{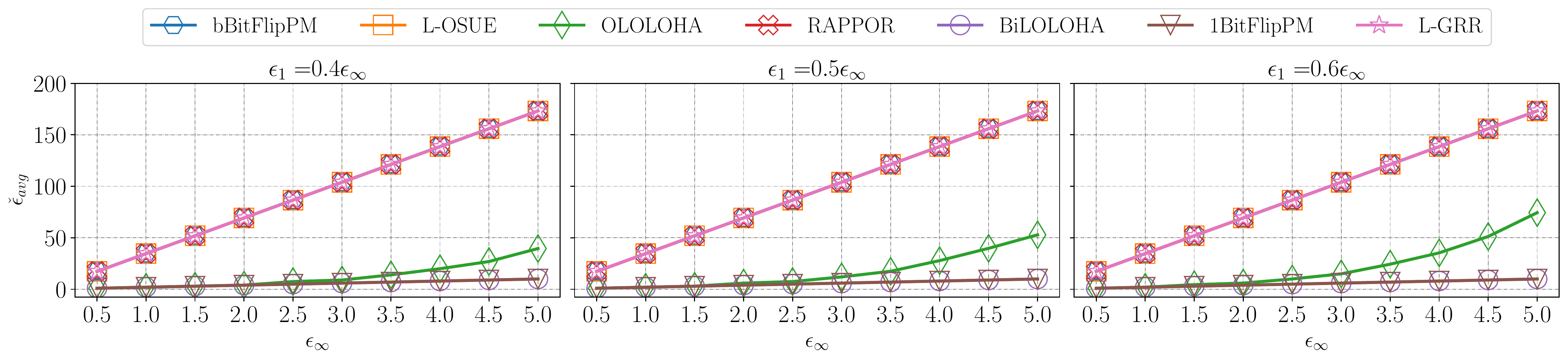}
  \caption{Adult dataset: $k=96$, $n=45222$, and $\tau=260$.}
  \label{subfig:eps_b}
\end{subfigure}
\\
\begin{subfigure}{2\columnwidth}
  \centering
  \includegraphics[width=0.96\linewidth]{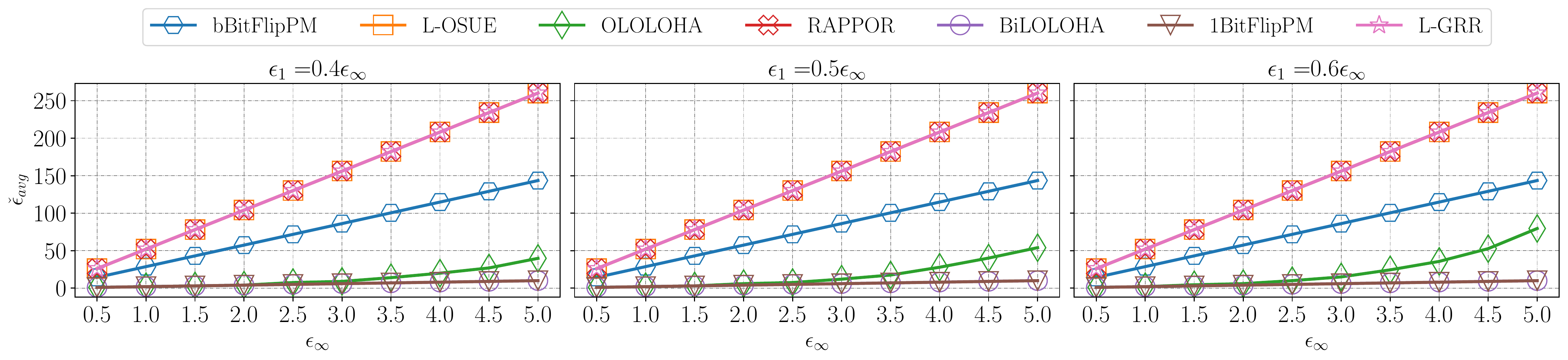}
  \caption{DB\_MT dataset: $k=1412$, $n=10336$, and $\tau=80$.}
  \label{subfig:eps_c}
\end{subfigure}%
\\
\begin{subfigure}{2\columnwidth}
  \centering
  \includegraphics[width=0.96\linewidth]{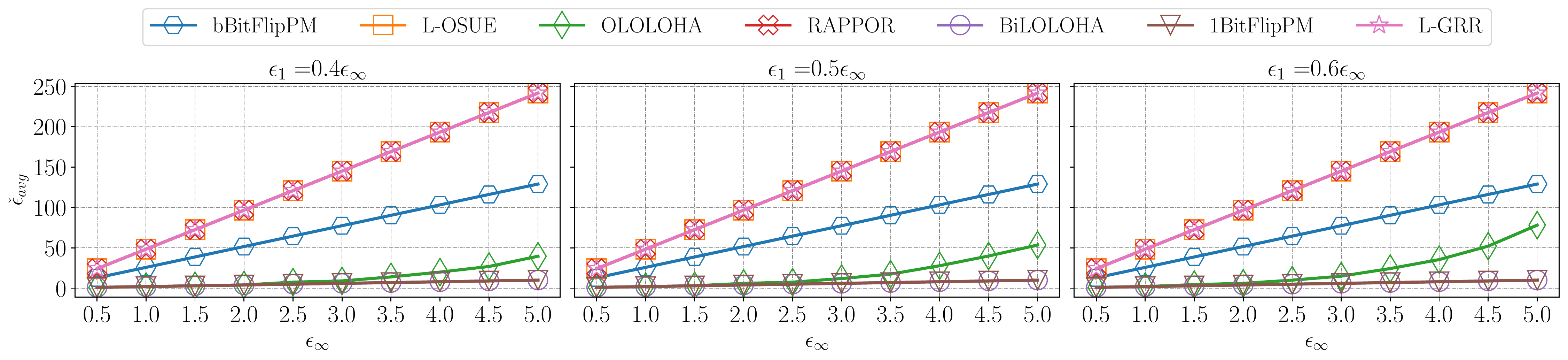}
  \caption{DB\_DE dataset: $k=1234$, $n=9123$, and $\tau=80$.}
  \label{subfig:eps_d}
\end{subfigure}
\caption{Averaged longitudinal privacy loss for $\tau$ data collections in Eq.~\eqref{eq:eps_avg} by varying the longitudinal $\epsInf$ and first report $\epsOne=\alpha \epsInf$ privacy guarantees, for $\alpha \in \{0.4, 0.5, 0.6\}$, on (a) Syn, (b) Adult, (c) DB\_MT, and (d) DB\_DE datasets. 
The evaluated methods are: $d$BitFlipPM~\cite{microsoft}, L-OSUE~\cite{Arcolezi2021_allomfree}, RAPPOR~\cite{rappor}, L-GRR~\cite{Arcolezi2021_allomfree}, and our LOLOHA protocols.}
\label{fig:eps_results}
\end{figure*}

From Fig.~\ref{fig:eps_results}, one can remark that all four LDP protocols, RAPPOR, L-OSUE, L-GRR, and $b$BitFlipPM (when $b=k$ in Figs.~\ref{subfig:eps_a} and~\ref{subfig:eps_b}), have an averaged longitudinal privacy loss linear to the number of data changes the users performed throughout the $\tau$ data collections.
Fig.~\ref{subfig:eps_a} presents the smallest $\check{\epsilon}_{avg}$ as both $k=360$ and the change rate $p_{ch}=0.25$ are small. 
However, in a worst-case scenario in which the users change their values significantly or $\tau \rightarrow \infty$, the overall privacy loss of RAPPOR, L-OSUE, L-GRR, and $b$BitFlipPM can grow to values as large as $k \epsInf$ for all datasets.
Note that in Figs.~\ref{subfig:eps_a} and~\ref{subfig:eps_b}, naturally, setting $b=k$ does not benefit from the $d$BitFlipPM advantage for enhancing longitudinal privacy protection by mapping several close values to the same bin, which leads to higher $\check{\epsilon}_{avg}$.
In contrast, in Figs.~\ref{subfig:eps_c} and~\ref{subfig:eps_d}, the longitudinal privacy loss of $b$BitFlipPM protocols is lower than RAPPOR, L-OSUE, and L-GRR because $b=\lfloor k/4 \rfloor$ buckets, but still significantly higher than our LOLOHA protocols. 

Indeed, the privacy loss of our LOLOHA protocols depends only on the new domain size $g \geq 2$, which is agnostic to $k$.
For this reason, our BiLOLOHA protocol with $g=2$ leaked about 15 to 25 orders of magnitude less than the state-of-the-art LDP protocols considering the experimented $\tau$ values.
These are similar results achieved by the $1$BitFlipPM protocol, which agrees with the theoretical analysis in Table~\ref{table-theory}, although BiLOLOHA consistently and considerably outperforms $1$BitFlipPM in terms of utility loss (see Fig.~\ref{fig:mse_results}).
Besides, since our OLOLOHA protocol has privacy loss depending on the optimal $g$ value in Eq.~\eqref{eq:opt_g}, which can be $g>2$ in low privacy regimes, it only resulted in about 2 to 5 order of magnitude less privacy loss than the state-of-the-art LDP protocols, for the experimented $\tau$ value. 
More specifically, when $\epsInf$ is high and $\alpha=0.6$ (see Fig.~\ref{subfig:eps_d}), OLOLOHA leaked about 2 orders of magnitude less privacy loss than the $b$BitFlipPM protocol.
However, as the number of data collections $\tau \rightarrow \infty$, $b$BitFlipPM privacy loss will go to $b\epsInf$, which is $b/g$ times higher than the one from OLOLOHA with $g\epsInf$.
Besides, in practice, lower values of $\epsInf$ and $\alpha \ll 1$ should be used to ensure strong longitudinal privacy guarantees since the first $\epsOne=\alpha \epsInf$-LDP report.
As shown in Fig.~\ref{fig:analysis_g}, this will mean lower values of $g$, which will substantially decrease the longitudinal privacy loss of OLOLOHA.

\subsection{Discussion}  
\label{sub:discussion}

In brief, we have evaluated in our experiments the performance of our LOLOHA protocols in comparison with four state-of-the-art memoization-based LDP protocols~\cite{rappor,microsoft,Arcolezi2021_allomfree} for frequency monitoring on different datasets and varying different parameters.
We now summarize the main findings that help justify the many claims of our paper.

More precisely, the conclusions we stated in Section~\ref{sec:theoretical_comparison} are based on the analytical variances of the LDP protocols. 
To corroborate these conclusions, our empirical experiments in Section~\ref{sub:results}, which measured the MSE metric, do indeed correspond to the numerical results of the variances.

Furthermore, the main disadvantage of RAPPOR~\cite{rappor} and the two others optimized protocols from~\cite{Arcolezi2021_allomfree}, (\ie, L-GRR and L-OSUE), is the linear relation on $k$ for the overall longitudinal privacy loss, \ie, $k\epsInf$, as each data change needs to be memoized.
Thus, for the monitoring of large-scale systems (\eg, application usage, calories ingestion, preferred webpage, etc), the overall privacy loss of such protocols will be tremendous, being unrealistic for private frequency monitoring. 

Even though the $d$BitFlipPM~\cite{microsoft} generalizes the original domain size $k$ to $b$ buckets, there is still a linear relation on the new domain size $b \leq k$ for the overall longitudinal privacy loss, \ie, $b\epsInf$, as each bucket change needs to be memoized \textit{when the mechanism is tuned for utility}.
What is more, this generalization naturally leads to loss of information and one has to carefully choose the bucket numbers/width for the best privacy-utility trade-off.
Besides, the privacy-utility trade-off of $d$BitFlipPM also depends on the number of bits $d\leq b$ each user samples. 
However, even when $d=1$, which offers the strongest protection~\cite{microsoft}, in our experiments, the server was still able to detect \textbf{all} bucket change of a small portion of users (see Table~\ref{tab:change_detection}).
Hence, as one adjusts $d$ for utility, \ie, $1<d \leq b$, the higher the attacker's success rate to detect all user's data changes will be.

The best choice for adequately balancing privacy and utility for frequency monitoring is with our LOLOHA protocols, as the privacy loss is only linear to the new (reduced) domain size $2\leq g \ll k$.
Though we only experiment with $80 \leq \tau \leq 260$ data collections in Section~\ref{sub:results}, in the worst case, this represents a significant $k/g$ decrease factor of privacy loss by our LOLOHA protocols. 
Intuitively, LOLOHA can be tuned to satisfy the strongest longitudinal privacy protection by selecting $g=2$ (\ie, our BiLOLOHA protocol).
In this setting, there is loss of utility in the encoding step through local hashing since the output is just one bit. 
For instance, even if this bit is transmitted correctly after the two rounds of sanitization, the server can only obtain one bit of information about the input (\ie, to which half of the input domain the value belongs to~\cite{tianhao2017}).
Nevertheless, from the analytic variance analysis in Fig.~\ref{fig:analysis_var} and empirical experiments in Fig.~\ref{fig:mse_results}, LOLOHA is optimal with $g=2$ in high privacy regimes, \ie, low $\epsInf$ values, which is desirable for practical deployments.

As a limitation, users fix their randomly selected hash function $\rmH \in \mathscr{H}$ with our LOLOHA protocols (\cf{}  Algorithm~\ref{alg:loloha}), which can be regarded as a unique identifier in longitudinal data collection. 
However, this is a common assumption of the LDP model, which assumes the server already knows the users' identifiers~\cite{Bittau2017,wang2019,erlingsson2019amplification,erlingsson2020encode}, but not their private data.
One way to counter this link between the user's randomized report and their identifier is to assume a trusted intermediate, such as a shuffler, that does not collude with the server, \eg, the Shuffle DP model~\cite{Bittau2017,erlingsson2019amplification,erlingsson2020encode}, which we let the investigation for future work.

\section{Related Work}  \label{sec:rel_work}

Differential privacy~\cite{Dwork2006,Dwork2006DP,dwork2014algorithmic} has been increasingly accepted as the current standard for data privacy. 
The central DP model assumes a trusted curator, which collects the clients' raw data and releases sanitized aggregated data.
The LDP model~\cite{first_ldp,Duchi2013,Duchi2013_b} does not rely on collecting raw data anymore, which has a clear connection with the concept of randomized response~\cite{Warner1965}. 
In recent years, there have been several studies on the local DP setting, \eg, for frequency estimation of a single~\cite{tianhao2017,Hadamard,Feldman2022,kairouz2016discrete,kairouz2016extremal,Naor2020,Cormode2021} and multiple~\cite{Arcolezi2021_rs_fd,Varma2022,Liu2023} attributes; mean estimation~\cite{xiao2,wang2019}, heavy hitter estimation~\cite{Bassily2015,bassily2017practical}, and machine learning~\cite{Chamikara2020,zhou2021local}.

As for locally differentially private monitoring, Erlingsson, Pihur, and Korolova~\cite{rappor} proposed the RAPPOR algorithm for frequency monitoring that is based on the \textit{memoization} solution described in Section~\ref{sub:long_fo}. 
The recent study of Arcolezi \etal~\cite{Arcolezi2021_allomfree} generalizes this framework for optimally chaining two LDP protocols, proposing the L-GRR protocol that is optimized for small domain size $k$ and the L-OSUE protocol for higher $k$ (see Figs.~\ref{fig:analysis_var} and~\ref{fig:mse_results}). 
Moreover, Erlingsson \etal~\cite{erlingsson2020encode} formalize the privacy guarantees of using two rounds of sanitization under both local and shuffle DP guarantees. 
Naor and Vexler~\cite{Naor2020} also formalized the privacy guarantees of chaining two LDP protocol as well as introduced a new Everlasting privacy definition.

An alternative approach for memoization named $d$BitFlipPM has been proposed by Ding, Kulkarni, and Yekhanin~\cite{microsoft}, discussed in Section~\ref{sub:microsoft}.
The $d$BitFlipPM protocol allows frequent but only \textit{small} changes in the original data since a high change (\ie, a different bucket) can be detected by an attacker (\cf{} Table~\ref{tab:change_detection}). 
Although an attacker that is able to identify a data change can still not infer the user's actual data (controlled by $\epsInf$), the overall LDP guarantees can be highly reduced if these changes are correlated~\cite{rappor,microsoft,erlingsson2019amplification}.
For instance, the authors in~\cite{tang2017privacy} performed a detailed analysis of Apple's LDP implementation and examined its longitudinal privacy implications. Naor and Vexler~\cite{Naor2020} also investigated the trackability of RAPPOR following their new Everlasting privacy definition.

LOLOHA leverages the best of RAPPOR and $d$BitFlipPM, which can inherently minimize these inference attacks.
More precisely, on the one hand, LOLOHA uses LH~\cite{tianhao2017} for domain reduction, which allows many values to collide (universal hashing property) and thus creates uncertainty about the user's actual value. 
Indeed, LH protocols are the least attackable LDP protocols in the recent studies of Arcolezi \etal~\cite{arcolezi2022risks} and Emre Gursoy \etal~\cite{Gursoy2022} considering a Bayesian adversary.
Besides that, LOLOHA also has two rounds of sanitization following RAPPOR's framework, which can improve privacy to minimize data change detection. 
Finally, another line of work for frequency monitoring under LDP is data change-based~\cite{Joseph2018,erlingsson2019amplification,Xue2022,Ohrimenko2022}, motivated by the fact that, generally, users’ data changes infrequently. 
A similar idea was proposed much earlier in the work of Chatzikokolakis, Palamidessi, and Stronati~\cite{Chatzikokolakis2014}, which proposed a predictive mechanism for location-based systems to utilize privacy budget only for new ``hard'' location points (\ie, with bad predictions).
However, these approaches normally impose restrictions on the number of data collections $\tau$ and on the number of data changes as their accuracy degrades linearly or sub-linearly with the number of changes in the underlying data distributions, which can limit their applicability and scalability to real-world systems.

\section{Conclusion and Perspectives}  \label{sec:conclusion}

In this paper, we study the fundamental problem of monitoring the frequency of evolving data throughout time under LDP guarantees. 
We proposed a new locally differentially private protocol named LOLOHA, which is built on top of domain reduction to minimize longitudinal privacy loss up to a $k/g$ factor and double randomization to enhance privacy. 
Through theoretical analysis, we have proven the longitudinal privacy (Theorems~\ref{theo:prr_loloha},~\ref{theo:irr_loloha}, and~\ref{theo:priv_loloha}) and accuracy guarantees (Proposition~\ref{prop:estimator_bounds}) of our LOLOHA protocols. 
In addition, through extensive experiments with synthetic and real-world datasets, we have shown that our proposed LOLOHA protocols preserve competitive utility as state-of-the-art LDP protocols~\cite{rappor,microsoft,Arcolezi2021_allomfree} by considerably minimizing longitudinal privacy loss (from 2 to 25 orders of magnitude with the experimented $\tau$ values). 
As future work, we intend to identify reasonable conditions of the input data (not constant as in~\cite{rappor,Arcolezi2021_allomfree}) in which one can satisfy the standard $\epsilon$-LDP definition.
Besides, we intend to identify attack-based approaches to longitudinal LDP frequency estimation protocols (\emph{e.g.}, data change detection or correlated data) and to extend the analysis of our LOLOHA protocols to the shuffle DP model.
Last, we also aim to integrate LOLOHA to the \texttt{multi-freq-ldpy} package~\cite{Arcolezi2022_multi_freq_ldpy}.

\begin{acks}
The authors deeply thank the anonymous meta-reviewer and reviewers for their insightful suggestions.
The work of Héber H. Arcolezi, Carlos Pinzón, and Catuscia Palamidessi was supported by the European Research Council (ERC) project HYPATIA under the European Union’s Horizon 2020 research and innovation programme. Grant agreement n. 835294. Sébastien Gambs is supported by the Canada Research Chair program as well as a Discovery Grant from NSERC.
\end{acks}

\bibliographystyle{ACM-Reference-Format}
\bibliography{B_references.bib}

\end{document}